\documentclass[11pt]{article}
\usepackage[default, probability style=bb, margin=1.25in, mathic=false]{ezlib}
\usepackage{csquotes}
\usepackage{tikz}
\usetikzlibrary{decorations.pathreplacing}
\MakeOuterQuote{"}

\selectcolormodel{natural}
\usepackage{ninecolors}
\selectcolormodel{rgb}

\usepackage{tabularray}
\UseTblrLibrary{booktabs}

\newcommand\extrafootertext[1]{%
    \bgroup%
    \renewcommand\thefootnote{\fnsymbol{footnote}}%
    \renewcommand\thempfootnote{\fnsymbol{mpfootnote}}%
    \footnotetext[0]{#1}%
    \egroup%
    \ignorespaces
}

\makeatletter
\@ifclassloaded{article}{
    \titleformat{\subsubsection}[runin]{\normalfont\ezlib@titlestyle}{\thesubsubsection}{\ezlib@headingsinline@numtitlesep}{\ezlib@headingsinline@period}
    \titlespacing*{\subsubsection}{0pt}{\medskipamount}{\ezlib@headingsinline@titletextsep}
    \titleformat{\paragraph}[runin]{\normalfont\itshape}{\thesubsubsection}{\ezlib@headingsinline@numtitlesep}{\ezlib@headingsinline@period}
    \titlespacing*{\paragraph}{0pt}{\medskipamount}{\ezlib@headingsinline@titletextsep}
}{}
\makeatother

\undef\see
\newcommand{\see}[2]{(#1{#2})}

\DeclareProbabilityCommand{\tfm}{P}

\newmathcommand{\Cchain}{C^{\text{chain}}}
\newcommand{\stocheq}{\mathrel{\stackrel{\scriptscriptstyle \mathrm{d}}{=}}}

\newcommand{\e}{\mathrm{e}}

\newcommand{\jobPair}{\calJ}

\newcommandPIE{\jjt}{\widetilde{\jobPair#1#2#3}}
\newcommand{\lst}{\widetilde}

\newrobustcmd{\shefali}[1]{{\color[HTML]{1591ea}[Shefali] \ignorespaces#1}}
\newrobustcmd{\edwin}[1]{{\color[HTML]{990099}[Edwin] \ignorespaces#1}}
\newrobustcmd{\ziv}[1]{{\color[HTML]{C8367C}[Ziv] \ignorespaces#1}}

\def\informalVersionOf\ref#1{%
    {informal version of \cref{#1}}%
}

\def\formalVersionOf\ref#1{%
    {formal version of \cref{#1}}%
}

\title{Priority Scheduling in the M/G/1\\with Preemption Overhead}
\author[Shefali Ramakrishna, Edwin Peng, and Ziv Scully.]{Shefali Ramakrishna\\Cornell University \and Edwin Peng\\Akuna Capital \and Ziv Scully\\Cornell University}

\begin{document}

\maketitle

\extrafootertext{This work was supported by the NSF under grant nos. CSR-1763701, CMMI-1938909, and CMMI-2307008.}

\begin{abstract}
Virtually all practical settings where preemptive scheduling is employed are susceptible to preemption overhead, and accounting for these overheads is necessary to make informed scheduling design decisions. However, preemption overhead is almost never accounted for in queueing-theoretic analyses of preemptive scheduling policies. This is true even for simple preemptive policies in simple queueing models: even the stability region, let alone the response time distribution, is difficult to analyze under overhead.

In this work, we give the \emph{first response time distribution analysis} of an M/G/1 under a preemptive scheduling policy with preemption overhead. Specifically, we consider class-based preemptive priority, where a stochastic overhead is incurred when pausing or resuming a job. We derive a recursive formula for the Laplace transform of response time for jobs of any given class, from which all response time moments can be extracted.

Beyond the specific policy and model we analyze, our broader aim is to provide a first step towards a general framework for analyzing queues with preemption overhead. To that end, we perform much of our analysis in a way that applies to a wide variety of overhead models by introducing a new theoretical tool called the \emph{job joint transform}.
\end{abstract}

\section{Introduction}

Preemptive scheduling policies, which allow pausing jobs mid-service, are ubiquitous because they allow important jobs to receive service ahead of unimportant jobs that would otherwise delay their completion. The canonical example is Shortest Remaining Processing Time (SRPT), which preemptively serves the job with least remaining work at every moment in time \citep{schrage_queue_1966}. There is a robust literature analyzing \emph{response time} (elapsed time between a job's arrival and completion) in the M/G/1 queue under many preemptive policies \citep{harchol-balter_performance_2013, scully_soap_2018, scully_new_2022}, shedding light on questions such as how preemption affects the mean and tail of response time, and whether preemption is unfair towards low-priority jobs.

In practice, there is a cost to preempting individual jobs, known as \emph{preemption overhead}. For instance, in computer systems, overhead can take the form of context switch, reloading from disk into memory, or reloading from memory into cache \citep{li_quantifying_2007}. However, with very few exceptions (\cref{prior-work}), the queueing theory literature does not account for preemption overhead. It is thus unclear how overhead affects preemption tradeoffs.

We give the \emph{first transform response time analysis of the M/G/1 queue under preemptive priority scheduling with preemption overhead}. Our overhead model is flexible: overhead amounts are stochastic and can occur whenever a job is paused, resumed, or both. We derive recursive formulas for the Laplace-Stieltjes transform of each priority class's response time distribution, which yield closed-form formulas for all moments of response time.

\subsection{Challenges and New Techniques}


Response time distributions for many priority scheduling policies in the M/G/1 have been analyzed over the years \citep{fajardo_waiting_2017, stanford_waiting_2014, kleinrock_queueing_1976, scully_soap_2018, scully_soap_2018a, wierman_nearly_2005, wierman_scheduling_2008, boxma_tails_2007, harchol-balter_performance_2013, schrage_queue_1966, schrage_queue_1967, osipova_optimal_2009, takacs_delay_1963, vreumingen_queueing_2019, bansal_achievable_2018, nunez-queija_queues_2002, wierman_is_2012, scully_characterizing_2020, lin_heavy-traffic_2011, kamphorst_heavy-traffic_2020, nair_tail-robust_2010}. Yet \emph{none} of this prior work incorporates preemption overhead. Why is it that preemption overhead has been left out despite occurring in the real world?

The main reason is that preemption overhead makes \emph{busy periods} difficult to analyze. Recall that a busy period started by a set of jobs is the total service time of this initial set of "level~0" jobs, "level~1" arrivals that occur during service of level~0 jobs, "level~2" arrivals during level~1 jobs, etc. Busy period transforms are central to many prior M/G/1 response time transform analyses as well as in the study of stability conditions \citep{harchol-balter_performance_2013, scully_soap_2018, kleinrock_queueing_1976, wierman_nearly_2005}. Under preemptive priority scheduling, preemption overhead creates a "two-way" dependency between service length and the number of arrivals during service. For example, when a high-priority job~$J_1$ arrives during service of a low-priority job~$J_2$ and incurs preemption overhead, $J_1$ effectively increases $J_2$'s service time, which in turn could cause more jobs to arrive during~$J_2$. This means that even stability, let alone response time, is hard to characterize under preemption overhead.


We handle the "two-way" dependency using a new theoretical tool we call the \emph{job joint transform}. This is a multivariate transform that captures the joint distribution of a job's effective service time (i.e., including overhead) and the arrivals during its service. Our main insights are that under Poisson arrivals, (1)~job joint transforms are tractable to compute, even for somewhat complex preemption overhead models; (2)~job joint transforms are sufficient to characterize busy period lengths, which is the primary obstacle to analyzing scheduling under preemption overhead; and (3)~job joint transforms unify a variety of other systems with two-way dependencies, such as preempt-repeat models \see\cref{sec:extensions}.

\subsection{Contributions and Outline}
In this paper, we:
\* Introduce the job joint transform for systems with service-arrival dependency \see\cref{sec:JJT-stability}.
\* Analyze busy periods and derive stability conditions for systems with service-arrival dependency \see\cref{sec:JJT-stability}.
\* Express a flexible model of preemptive priority with preemption model as a system with service-arrival dependence \see\cref{sec:JJT-stability}.
\* Derive the job joint transform for our model \see\cref{sec:jjt-derivation}.
\* Derive our response time results for a system with preemption overhead \see\cref{sec:response-time}.
\* Provide expressions for the job joint transform in a number of other service-arrival-dependent systems \see\cref{sec:extensions}.
\*/
We view our work as a first step towards incorporating preemption overhead into broader M/G/1 scheduling theory \see\cref{sec:conclusion}.
This work is the full version of previously published 2- or 3-page abstracts \citep{peng_exact_2022, ramakrishna_transform_2024}. 

\subsection{Prior Work}\label{prior-work}

\subsubsection{Queues with Switching Overheads}

Most literature on overheads in queues covers overheads from switching between \emph{classes} of jobs in a multiclass queue or, roughly equivalently, between \emph{queues} of jobs in polling systems \citep{van_oyen_optimality_1992, neuts_mg1_1977, boon_applications_2011, borst_polling_2018, cao_stability_2017, borst_polling_1997, wang_heavy_2021, vreumingen_queueing_2019, boxma_pseudo-conservation_1987, bertsimas_optimization_1999a}. This work does not capture preemption overhead because almost all the policies studied never switch classes or queues \emph{during} a job's service. To the best of our knowledge, the only exception is the work of \citet{cao_stability_2017}, but they study a \emph{preempt-restart} model, meaning that preempted jobs do not retain progress, with only two job classes.

We are aware of just one prior paper studying scheduling in an M/G/1 with overhead caused by preemption of individual jobs: \citet{goerg_evaluation_1986} studies SRPT with overhead.
However, \citet{goerg_evaluation_1986} only characterizes mean response time, not the full Laplace transform, and has a less flexible preemption model than ours: overhead is deterministic and only occurs when pausing jobs, not resuming them.

\subsubsection{Queues with Service-Arrival Dependencies}\label{prior-work:service-arrival}
The presence of preemption overhead in preemptive-priority systems can be viewed as a complex interdependency between the service length of the preempted job and the number of preemptions resulting from higher-priority arrivals during that job's service. Several prior works have studied a variety of queueing systems where there is some sort of dependency between the arrival process and the service length of jobs. These works derive stability conditions, response times, and other properties of the system. The systems studied include:
\* \emph{Preempt-repeat systems:} Arrivals interrupt service and force jobs to restart under different scheduling policies, including PLCFS and preemptive priority \citep{drekic_reducing_2001, asmussen_preemptive-repeat_2017, bergquist_stationary_2022, schrage_mixed-priority_1969, avi-itzhak_preemptive_1963a, fajardo_waiting_2017, gaver_waiting_1962, chang_preemptive_1965}.
\* \emph{Multiclass queue with state-dependent arrival rates:} Arrival rates depend on the class of the job in service \citep{ernst_stability_2018}.
\*/
All of these systems share structural similarities but differ in the nature of their service-arrival dependencies. In preempt-repeat models, arrivals and preemptions directly increase a job’s service time, whereas in \citet{ernst_stability_2018}, the job currently in service influences the arrival process. Stability in these systems is typically analyzed using (multitype) Galton-Watson branching processes, though the details of the analysis differ from system to system.

Our work introduces a unifying technique that captures all of the above scenarios. The key analytical tool is what we call the \emph{job joint transform}, the multivariable transform of a job's service time and the number of arrivals of each class during said service time. The job joint transform captures all relevant service-arrival dependencies of the system. Our characterizations of the stability condition \see\cref{thm:stability} and busy period lengths \see\cref{thm:busy-period-transform} work for generic job joint transforms, generalizing and extending prior work. In \cref{sec:extensions} we demonstrate how to frame prior work using job joint transforms.


\section{System Model}\label{model}

We consider a multi-class M/G/1 queue. Jobs of classes $\{1,\dots,n\}$ arrive according to independent Poisson processes of positive rates $\lambda_1, \dots, \lambda_n$, and job sizes (a.k.a. service times) for each class are drawn i.i.d. from positive distributions $S_1, \dots, S_n$. This job size $S_k$ excludes preemption overhead \see\cref{model:overhead}, and when we want to emphasize this fact, we call $S_k$ a job's \emph{original} size.

Throughout, we use the shorthand "class~$< k$" to mean "of some class in $\{1, \dots, k - 1\}$", and similarly for $\leq k$, $> k$, and $\geq k$. See also our notation conventions in \cref{model:notation}.

The system uses the \emph{preemptive priority} scheduling policy. This means the server prioritizes working on the job with the lowest class index, i.e. for $i < j$, class~$i$ has priority over class~$j$. Within each class, jobs are served in first-come first-served (FCFS) order. We assume a preempt-resume model, meaning that when a job is preempted, all progress made on it so far is retained.

We write $T_k$ for the response time distribution of class~$k$ jobs. Our main result \see\cref{thm:response-time} is an exact characterization of $\lst{T}_k(\theta) = \E{e^{-\theta T_k}}$, the Laplace-Stieltjes transform of~$T_k$.

\subsection{Preemption Overhead}\label{model:overhead}

Compared to a typical work-conserving M/G/1 model, we add two sources of overhead related to preemption.
\* Whenever a class~$k$ job in service is preempted to make way for a class~$< k$ job, the arriving job triggers \emph{class~$k$ pause overhead}, or simply a "class~$k$ pause".
\* Whenever a previously preempted class~$k$ job is about to enter service again, there is a \emph{class~$k$ resume overhead}, or simply a "class~$k$ resume".
\*/
Both types of overhead are nonpreemptible, so once an overhead starts, no job can enter service until the overhead ends.

The length of each overhead is drawn i.i.d. from a distribution that depends on (a)~whether the overhead is a pause or resume and (b) the class of the job being paused or resumed.
Class~$k$ pauses have length distribution~$C_k$, and class~$k$ resumes have length distribution~$D_k$.


Class-based priority and preemption overhead could in principle interact in many ways.
We assume they interact as described in \cref{alg:pprio_with_overhead}.
In brief, we assume that preemption overheads are nonpreemptible.
But our methods are not specific to the details of \cref{alg:pprio_with_overhead}, as we demonstrate in \cref{sec:extensions}.

There are two key features to notice about the interaction between class-based priority and preemption overhead \see\cref{alg:pprio_with_overhead}:
\* When a class~$k$ job is preempted due to the arrival of a class~$< k$ job, after the resulting pause overhead, the job of best priority enters service. This may or may not be the job whose arrival caused the pause in the first place.
\* When a class~$k$ job is entering service again, if any class~$<k$ jobs arrive during the resulting resume overhead, the class~$k$ job is preempted again immediately after the resume. This triggers another pause immediately after the resume.
\*/
If class~$< k$ arrivals occur during a class~$k$ resume, we say the resume \emph{fails} or describe it as \emph{failed}; otherwise, we say it \emph{succeeds} or describe it as \emph{successful}. Only successful resumes result in the class~$k$ job resuming service towards its original size.


\begin{algorithm}[t]
\caption{Preemptive priority with pause and resume overhead}
\label{alg:pprio_with_overhead}
\small
\setlist[ezlist,2]{label={$\triangleright$}}
\setlist[ezlist,3]{label={$\triangleright$}}
\*[(a), left=0em] Whenever a job arrives to an empty system:
\** Begin serving the job
\* Whenever a class~$k$ job in service is preempted by a class~$< k$ arrival:
\** Stop serving the class~$k$ job
\** Begin a class~$k$ pause
\* Whenever a job or pause completes, as long as there are still jobs in the system:
\** Identify the best-priority job (least class index first, FCFS within each class)
\**[$\diamond$] Is the best-priority job a previously paused class~$k$ job?
\*** If so, begin a class~$k$ resume
\*** Otherwise, begin serving the best-priority job
\* Whenever a class~$k$ resume completes:
\**[$\diamond$] Did any class~$< k$ jobs arrive during the resume?
\*** If so, we say the resume \emph{fails}; begin a class~$k$ pause
\*** Otherwise, we say the resume \emph{succeeds}; begin serving the resumed class~$k$ job
\*/
\end{algorithm}

\subsection{Effective Size and the Job Joint Distribution}\label{model:effective_size}

We consider class~$k$ overheads to be part of the class~$k$ being paused or resumed, and we define a job's \emph{effective size} to be its total service time including overhead, i.e. its original size plus the total length of all of its overheads.
We can make this choice because \emph{overhead occurs only in response to a lower-class job arriving while a higher-class job is being served}.

The overhead dynamics in \cref{alg:pprio_with_overhead}, the arrival rate $\lambda_{< k}$ of class~$< k$ jobs \see\cref{model:notation}, and the distributions $S_k$, $C_k$, and $D_k$ together determine a \emph{joint distribution} $\gp[\big]{R_k, \vec{A}_k}$ on the following quantities (see also \cref{sec:JJT-stability:jjt}):
\* The \emph{effective size} of a class~$k$ job, denoted~$R_k$.
\* For each class~$i$, the \emph{number of class~$i$ arrivals} that occur during a class~$k$ job, denoted~$A_{k, i}$. We write $\vec{A}_k = [A_{k, i}]_{i = 1}^n = [A_{k, 1}, \dots, A_{k, n}]$ for the vector of these numbers of arrivals.
\*/
To clarify, the arrival counts include both arrivals during the job's original work and arrivals during those overheads. Occasionally, we further subdivide $R_k$ as
\[
    R_k = S_k + C^*_k + D^*_k,
\]
where $C^*_k$ and $D^*_k$ are the total length of all pause and resume overheads, respectively, that occur during a class~$k$ job.

We characterize the joint distribution $\gp{R_k, \vec{A}_k}$ in \cref{lem:JJT-complete}. This joint distribution, which we call the \emph{job joint distribution}, is important for our analysis \see\cref{sec:JJT-stability}. We do much of our reasoning about stability and busy period lengths using job joint distributions and other "\emph{service-arrival joint distributions}" of this form.

    One might be concerned that the joint distribution $(R_k, \vec{A}_k)$ might not be well defined in the sense that it might depend on exactly what times the job is in service.
    Fortunately, our model has time-homogeneous Poisson arrivals, so these are well defined (see \cref{sec:local-arrivals}).

\subsection{Types of Load}
\label{model:loads}
We define four types of loads for each class~$k$: one for original work, two for overhead work, and one covering all work.
\* The \emph{class~$k$ original load} is $\sigma_k = \lambda_k \E{S_k}$.
\* The \emph{class~$k$ pause load} is $\gamma_k = \lambda_k \E{C^*_k}$.
\* The \emph{class~$k$ resume load} is $\delta_k = \lambda_k \E{D^*_k}$.
\* The \emph{class~$k$ effective load} is $\rho_k = \sigma_k + \gamma_k + \delta_k = \lambda_k \E{R_k}$.
\*/
We can interpret $\sigma_k$ as the rate at which class~$k$ original work arrives to the system.
Equivalently, if the system is stable, $\sigma_k$ is the fraction of time that the server is in class~$k$ original mode.
The other loads have analogous interpretations.

\subsection{Notation}\label{model:notation}

\newcommand{\BigNotationTable}{%
    \begin{booktabs}{@{}lX>{\crefShorten}l@{}}
        \toprule
        \scshape Notation & \scshape Description & \scshape Defined in \\
        \midrule
        $n$ & number of classes & \cref{model} \\
        $< k$, $\leq k$, $> k$, $\geq k$ & shorthand for sets of classes, e.g. $> k$ is $\{k + 1, \dots, n\}$ & \cref{model} \\
        $\stocheq$ & equality in distribution & \cref{model:notation} \\
        $\lst{V}(\theta) = \E{e^{-\theta V}}$ & Laplace-Stieltjes transform of~$V$ at~$\theta$ & \cref{model:notation} \\
        $V_\e$ & Excess distribution of~$V$ & \cref{def:excess} \\
        \midrule
        $S_k$ & original size of a class~$k$ job (excludes overheads) & \cref{model} \\
        $C_k$ & length of a single class~$k$ pause & \cref{model:overhead} \\
        $D_k$ & length of a single class~$k$ resume & \cref{model:overhead} \\
        $C^*_k$ & total length of all pauses of a class~$k$ job & \cref{model:overhead} \\
        $D^*_k$ & total length of all resumes of a class~$k$ job & \cref{model:overhead} \\
        $R_k = S_k + C^*_k + D^*_k$ & effective size of a class~$k$ job (includes overheads) & \cref{model:overhead} \\
        $T_k$ & response time of a class~$k$ job & \cref{model} \\
        \midrule[dashed]
        $S_{< k}$, $R_{\leq k}$, $T_{> k}$, etc. & mixture of classes, e.g. $R_{\leq k}$ is $R_i$ w.p. $\frac{\lambda_i}{\lambda_{\leq k}}$ for $i \in \{1, \dots, k\}$ \\
        \midrule
        $\lambda_k$ & arrival rate of class~$k$ jobs & \cref{model} \\
        $\sigma_k = \lambda_k \E{S_k}$ & original load for class~$k$ & \cref{model:loads} \\
        $\gamma_k = \lambda_k \E{C^*_k}$ & pause load for class~$k$ & \cref{model:loads} \\
        $\delta_k = \lambda_k \E{D^*_k}$ & resume load for class~$k$ & \cref{model:loads} \\
        $\rho_k = \sigma_k + \gamma_k + \delta_k$ & effective load of class~$k$ & \cref{model:loads} \\
        \midrule[dashed]
        $\lambda_{< k}$, $\sigma_{\leq k}$, $\rho_{> k}$, etc. & sum over classes, e.g. $\rho_{> k} = \sum_{i > k} \rho_i = \sum_{i = k + 1}^n \rho_i$ \\
        $\lambda = \lambda_{\leq n}$, $\rho = \rho_{\leq n}$ & total arrival rate and effective load \\
        \midrule
        $A_{k, \ell}$ & number of class~$\ell$ arrivals during a class~$k$ job & \cref{model:overhead} \\
        $\vec{A}_k = [A_{k, i}]_{i = 1}^n$ & arrival vector, i.e. $[A_{k, 1}, \dots, A_{k, n}]$ & \cref{model:overhead} \\
        \midrule[dashed]
        $A_{< k, \ell}$, $\vec{A}_{\leq k}$, etc. & first-slot inequality denotes mixture of classes (like $S_{< k}$) \\
        $A_{k, < \ell}$, etc. & second-slot inequality denotes sum over classes (like $\lambda_{< \ell}$) \\
        \midrule
        $\jobPair_k = (R_k, \vec{A}_k)$ & class~$k$ job joint distribution & \cref{def:jjt} \\
        $\jjt_k(\theta, \vec{z})$ & class~$k$ job joint transform & \cref{def:jjt} \\
        \midrule[dashed]
        $\jobPair_{< k}$, etc. & mixture of classes, e.g. $\jjt_{\leq k}(\theta, \vec{z}) = \sum_{i \leq k} \frac{\lambda_i}{\lambda_{\leq k}} \jjt_i(\theta, \vec{z})$ \\
        \midrule
        $B_{< k}(V)$ & length of class~$< k$ busy period started by $V$ work & \cref{def:class-busy-period} \\
        $B_{< k}(\jobPair)$ & length of class~$< k$ busy period started by job with joint distribution~$\jobPair$ & \cref{def:class-busy-period} \\
        $B = B_{\leq n}$ & full busy period length & \cref{def:class-busy-period}\\
        \midrule
        $\vec{z} = [z_i]_{i = 1}^n$ & argument to joint transforms \\
        \midrule[dashed]
        $z_{< k}$, etc. & mixture of classes, e.g. $z_{\geq k} = \sum_{i \geq k} \frac{\lambda_i}{\lambda_{\geq k}} z_i$ \\
        \bottomrule
    \end{booktabs}
}

See \cref{app:notation_table} for a summary of notation. A few conventions are worth highlighting to aid interpretation, particularly regarding the use of subscripts such as $< k$ and similar forms, which denote quantities "focused on class~$< k$ jobs." The precise meaning of these subscripts depends on the context:
\* For quantities representing loads, arrival rates, or arrival counts (e.g., $\lambda_{< k}$ and $\sigma_{\leq k}$), the subscript indicates a sum over the specified range of classes. For example, $
\lambda_{< k} = \sum_{i=1}^{k-1} \lambda_i.$
\* For distributions (and their transforms) related to a randomly selected job class (e.g., $S_{< k}$ and $R_{\leq k}$), the subscript represents a mixture distribution where a class~$i$ job is selected with probability proportional to $\lambda_i$. For example, the transform of $S_{\leq k}$ is $\lst{S}_{\leq k}(\theta) = \sum_{i=1}^k \frac{\lambda_i}{\lambda_{\leq k}} \lst{S}_i(\theta).$
\* In the context of joint transforms, which often use arguments $z_1, \hdots, z_n$, the subscript $z_{< k}$ denotes the mixture over classes. When the subscript is omitted, it is understood to correspond to $z_{\leq n}$. For instance, $z_{< k} = \sum_{i=1}^{k-1} \frac{\lambda_i}{\lambda_{< k}} z_i$.
\* Busy periods fall outside the above conventions. In these cases, the $< k$ subscript indicates a class~$< k$ busy period \cref{sec:tree:response-time}, a busy period consisting of only class~$< k$ jobs' arrivals and service.
\*/
These conventions have been selected for internal consistency and to yield intuitive expressions. For example, although we define $\sigma_{\leq k} = \sum_{i=1}^k \lambda_i \E{S_i}$ directly, it follows from the first two conventions that $\sigma_{\leq k} = \lambda_{\leq k} \E{S_{\leq k}}$.
For the total arrival rate and total effective load, we write simply $\lambda = \lambda_{\leq n}$ and $\rho = \rho_{\leq n}$.

Throughout the paper, definitions and results stated with $\leq k$ have a $< k$ analogue, and vice versa: one can simply plug $k \pm 1$.
The same holds for $> k$ and $\geq k$.
When stating definitions and results, we generally choose whichever strictness is most commonly used, but we also use the opposite without further comment.

We use random variable names and their distributions interchangeably; for instance, we write "distributed as $S_k$" to mean "distributed as the distribution of $S_k$."

Finally, we state some standard definitions and theorems used throughout the paper.

\begin{definition}
    For a nonnegative real-valued random variable $V$, the Laplace-Stieltjes transform of $V$ evaluated at $\theta \geq 0$ is defined as
    \[
    \lst V(\theta) = \E{e^{-\theta V}}.
    \]
\end{definition}

\begin{definition}
    We write $X \stocheq Y$ to denote that random variables $X$ and $Y$ are equal in distribution.
\end{definition}

\begin{definition}\label{def:excess}
The \emph{excess distribution} of a nonnegative real-valued random variable distributed as $V$, denoted by $V_\e$, is:
\[
\P*{V_\e < t} = \frac{1}{\E{V}}\int_0^t \P{V > x}dx,
\]
with Laplace-Stieltjes transform
\[
\lst V_e(\theta) = \frac{1- \lst V(\theta)}{\theta \E{V}}.
\]
This is also known as the forward/backward recurrence time of $V$.
\end{definition}

\begin{theorem}[Wald’s equation]\label{thm:wald}
Let $\{Y_i\}_{i\geq 1}$ be i.i.d. random variables with finite mean, and let $N$ be a stopping time
with respect to $\{Y_i\}$ satisfying $\E{N}<\infty$. Then
\[
\E*{\sum_{i=1}^{N} Y_i} = \E{N}\E{Y_1}.
\]
\end{theorem}

\section{Stability and Busy Period Lengths with Overhead}\label{sec:JJT-stability}

In this section, we prove stability results for our model with preemption overhead. Ideally, we would like for standard M/G/1 stability results to hold for our model in a straightforward manner, but this is not possible. This is because preemption overhead introduces what we call \emph{service-arrival dependence}, in the sense that the amount of work required for a job to complete is not fixed in advance, because arrivals that occur while the job is in service may trigger overhead and thereby extend the time until the job completes. Consequently, standard M/G/1 arguments that condition on a fixed service length $s$ and then treat the number of arrivals during that service as an independent Poisson random variable distributed as $\mathrm{Poisson}(\lambda s)$ do not apply in our setting.

Our approach to deal with this service-arrival dependency is to represent a busy period as a branching object that records \emph{which jobs arrive during the effective service time of which other jobs}. This leads to a multitype Galton-Watson representation of the set of jobs served during a busy period. To analyze such branching structures in the presence of service-arrival dependence, we introduce a primitive, the \emph{job joint distribution}, which captures the joint law of a job's effective service time and the vector of arrivals that occur during it.

This section proceeds as follows. In \cref{sec:JJT-stability:trees} we define busy period trees and explain how they encode the recursive structure and service-arrival dependency of busy periods. In \cref{sec:JJT-stability:jjt} we formalize the service-arrival joint distribution and the joint Laplace-Stieltjes transform of service legth and number of arrivals, called the \emph{job joint transform}. In \cref{sec:JJT-stability:bp} we use the job joint transform to derive fixed-point equations for busy period transforms, and we extract stability conditions and formulas for the expected lengths of busy periods and other related values. In \cref{sec:JJT-stability:assumptions} we describe the extent to which our results generalize to cover other systems with service-arrival dependency, including other overhead models. Finally, we collect proofs and technical conventions in \cref{sec:JJT-stability:proofs}.

\subsection{Busy Period Trees}\label{sec:JJT-stability:trees}

A \emph{full busy period}, denoted $B$, is the amount of time between two consecutive epochs during which the server is empty (i.e., the elapsed time between when a job arrives to an empty server and when the server next becomes empty). As usual for M/G/1-type queueing systems, to show the system is stable, we would like to show the idle state is positive-recurrent. Therefore, it suffices to show that the length distribution of a fully busy period satisfies $\E{B} < \infty$. Proving this suffices because it implies a renewal process structure, with each full busy period and subsequent idle period forming a renewal cycle of finite length $\E{B} + 1/\lambda$. By Wald's Equation\footnote{This argument additionally requires that the number of jobs in a busy period is a stopping time and that $\E{R_k} < \infty$ for all classes $k$, both of which hold in our model.} \see\cref{thm:wald}, showing stability reduces to showing that the expected number of jobs served during a full busy period is finite.

We show this by analyzing a busy period as a multi-type Galton Watson branching process. We do this by describing the set of jobs served during a busy period using a tree representation that records \emph{which jobs arrive during the effective service of which other jobs}.  We formalize this notion below.

\begin{figure}
    \centering
    \includegraphics[page = 1, width=\linewidth, trim= 40 0 40 10, clip]{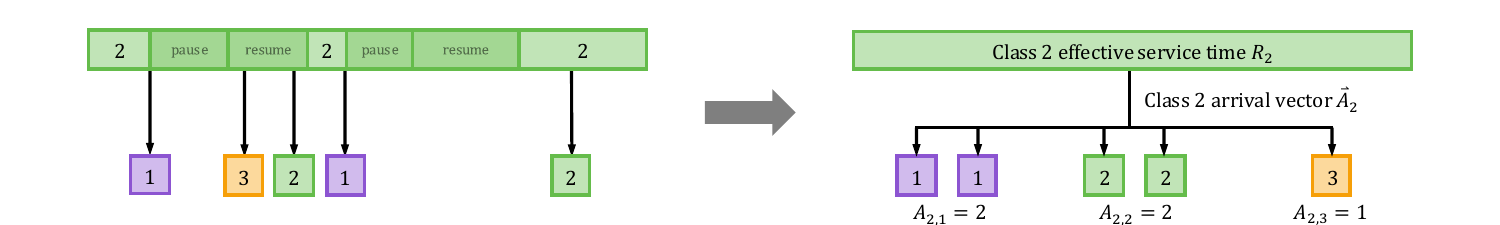}
    \caption{
    (Left)~The service history of a class~$2$ job represented via a busy period tree. Shaded regions mark periods of overhead. Arrows mark exactly when arrivals occurred during the root job's effective service time.
    (Right)~The corresponding realization of the job joint distribution $\calJ_2 = (R_2, \vec{A}_2)$
    \label{fig:multi-type}}
\end{figure}

\begin{definition}
A \emph{busy period tree} is a directed, rooted tree in which nodes of the tree correspond to the jobs served during the busy period, and each node is labeled by the class in $\{1,\ldots,n\}$ of the corresponding job.

For any two jobs $x$ and $y$ served during the busy period, there is a directed edge $x \to y$ if job $y$ arrived during the effective service time of job $x$, i.e., during the time interval over which $x$ receives its own service together with any pause and resume overhead attributed to $x$ being preempted.
The \emph{root} of the tree is the job whose arrival initiates the busy period.
\end{definition}

The ``busy period tree'' tree perspective is convenient for our analysis for two main reasons. First, it is a multitype Galton-Watson branching process, which allows us to apply existing tools to analyze its extinction and total progeny, which will yield the stability conditions we need. Second, this tree representation helps us reason in terms of "class-restricted" trees when analyzing response time \see\cref{sec:JJT-stability:restricted}.

\subsection{Characterizing Busy Periods Using the Job Joint Distribution}\label{sec:JJT-stability:jjt}

Now that we have described how the set of jobs served during a full busy period can be represented as a multitype branching process, we make this description more precise by explaining exactly how a busy period tree is generated in \ref{alg:busy_period_tree}.

\begin{algorithm}[t]
\caption{Generation of a busy period tree}
\label{alg:busy_period_tree}
\small
\setlist[ezlist,1]{label={}, wide}
\setlist[ezlist,2]{label={$\triangleright$}}
\setlist[ezlist,3]{label={$\triangleright$}}
\* Generate root of the busy period tree:
\** Assign the root class $k \in \{1, \hdots, n\}$ with probability $\lambda_k/\lambda$.
\* Generate offspring recursively for each node:
\** For a node of class~$k$, sample the service-arrival joint distribution $(R_k, \vec{A}_k) = (R_k, A_{k,1}, \hdots, A_{k,n})$.
\** For each class $i \in \{1, \hdots, n\}$, create $A_{k,i}$ class~$i$~offspring of the parent class~$k$ node.
\** For each offspring node, repeat the offspring generation process independently.
\*/
\end{algorithm}



Because arrivals are homogeneous Poisson processes and the scheduling policy is nonanticipative, all class~$k$ jobs have the same arrival vector distribution, and arrival vectors associated with different nodes in the busy period tree are independent conditional on their classes. It follows that the collection of jobs served during a full busy period forms a multitype Galton-Watson branching process with types $\{1,\ldots,n\}$ and offspring distribution for a type-$k$ node given by $\vec{A}_k$.

Our observation about branching processes already suffices to characterize stability. Since the number of jobs in a busy period is a stopping time and the expected service time of each job is finite, Wald’s equation \see\cref{thm:wald} implies that finiteness of the expected number of jobs in a busy period implies finiteness of the expected busy period length. By standard multitype Galton-Watson theory \citep[Chapter~V]{athreya_branching_1972}, extinction of the branching process (equivalently, finiteness of the total expected number of jobs in a busy period) is fully determined by the largest eigenvalue of the \emph{mean offspring matrix}
\[
M = \begin{bmatrix}
    \E{A_{1, 1}} & \E{A_{1,2}} & \hdots & \E{A_{1,n}} \\
    \E{A_{2,1}} & \E{A_{2,2}} & \hdots & \E{A_{2,n}} \\
    \vdots & \vdots & \ddots & \vdots \\
    \E{A_{n,1}} & \E{A_{n,2}} & \hdots & \E{A_{n,n}}
    \end{bmatrix}.
\]
Thus, the arrival vectors $\vec{A}_k$ alone are sufficient to characterize stability.

However, to characterize the exact \emph{distribution} of busy period lengths, which we need for response-time analysis, the arrival vectors are no longer sufficient. We need the joint distribution of a job’s effective service time and its arrival vector, not just their respective marginals. Because of the two-way dependency between these values, it is not quite obvious at first glance what this distribution should be. In the rest of this section, we will reduce the key questions to computing this distribution, and actually compute it in \cref{sec:jjt-derivation}.

In the standard M/G/1 queue without overhead, busy period transforms are defined recursively from the Laplace-Stieltjes transform of the service time distribution, since arrivals during a fixed service interval are conditionally Poisson. In our setting, an analogous recursive description still holds, but it must be expressed in terms of a joint transform that captures the service-arrival dependence.

Accordingly, we introduce the notion of a \emph{service-arrival joint distribution}, which specifies a joint distribution of a nonnegative duration and the vector of arrivals observed during that duration, and its associated joint Laplace-Stieltjes transform.

\begin{definition}
\*[subenv, afterheading]
\*A \emph{service-arrival joint distribution} $\calV = (R_\calV, \vec{A}_\calV)$ is a joint distribution in the space $(\bbR_+,\bbZ_+^n)$.
Given a service-arrival joint distribution $\calV$, we write $R_\calV$ for the \emph{work duration component} and $\vec{A}_\calV = (A_{\calV,1}, \hdots, A_{\calV,n})$ for the \emph{arrival count vector component}.
In general, $R_\calV$ and $\vec{A}_\calV$ are dependent.
\* Given $\calV$, the \emph{service-arrival joint transform}, denoted $\lst{\calV}$, is
\[
\lst{\calV}(\theta,\vec{z}) \coloneqq \E*{e^{-\theta R_\calV}\prod_{i=1}^n z_i^{A_{\calV,i}}}.
\]
\*/
We use the following notation conventions for service-arrival joint distributions:
\* We use calligraphic font for service-arrival joint distributions (except for the shorthand introduced in \cref{def:poisson-sajd}).
\* Service-arrival joint distributions are added pointwise: $\calV + \calW = (R_\calV + R_\calW, \vec{A}_\calV + \vec{A}_\calW)$.
\*/
\end{definition}

The most important instance of the service-arrival joint distribution arises when the work duration component $R_\calV$ represents the effective service time of a class~$k$ job. We call this service-arrival joint distribution the \emph{class~$k$ job joint distribution}, which captures the joint distribution of the effective service time of a class~$k$ job and the number of arrivals of each class during its effective service.\footnote{See \cref{sec:JJT-stability:assumptions} for why this joint distribution is well defined.}

\begin{definition}\label{def:jjt}
\*[subenv, afterheading]
\* The \emph{class~$k$ job joint distribution}, $\jobPair_k = (R_k, \vec{A}_k)$ is the service-arrival joint distribution of the class~$k$ effective service time $R_k$ and its associated class~$k$ arrival vector  $\vec{A}_k = (A_{k,1},\hdots,A_{k,n})$.
\* The \emph{class~$k$ job joint transform} is $\jjt_k$, the service-arrival joint transform of $\jobPair_k$.
\*/
\end{definition}

Aside from the job joint distributions $\jobPair_k$, the service-arrival joint distributions we use most frequently are those that arise from Poisson arrivals without any overhead or other non-standard dependence.

\begin{definition}\label{def:poisson-sajd}
\* [subenv, afterheading]
\* Given a duration distribution $V$, the \emph{standard Poisson joint distribution}, denoted $\calP(V)$, has $R_{\calP(V)} = V$ and $(A_{\calP(V), i}\given V) \sim \text{Poisson}(\lambda_i V)$, where all $A_{\calP(V), i}$ variables are mutually conditionally independent given~$V$.
\* The \emph{standard Poisson joint transform} is $\lst{\calP(V)}$, the service-arrival joint transform of $\calP(V)$.
\*/

Because $\calP$ is the "default" way to turn a duration distribution into a service-arrival joint distribution, we typically shorten $\calP(V)$ to simply $V$ to reduce clutter, as in $B(V + \jobPair) = B(\calP(V) + \jobPair)$.
\end{definition}

We give a formula for the standard Poisson joint transform in \cref{lem:poisson-standard-joint-tfm}.
It appears in some form throughout the queueing literature, though usually not explicitly stated, so we include the proof there for completeness.

In the next section, we will use job joint transforms to derive fixed-point equations for busy period transforms.

\subsection{Busy Period Lengths}\label{sec:JJT-stability:bp}

We now use the service-arrival joint distributions and transforms introduced in \cref{sec:JJT-stability:jjt} to characterize busy period lengths. We first analyze \emph{full} busy periods, which yield stability conditions and mean busy period lengths. Later in this section, we analyze "class-restricted busy periods", which we will need for response-time analysis.

A busy period can be viewed as the total amount of work generated, either directly or indirectly, by an initial amount of work. In systems with service-arrival dependence, this includes both the effective service time of the initiating job and all work generated by arrivals during that service, recursively.

\begin{definition}\label{def:class-busy-period}
A \emph{busy period started by a service-arrival joint distribution $\calV$},
denoted $B(\calV)$, is the distribution satisfying
\[
B(\calV) \stocheq
R_\calV + \sum_{i=1}^n \sum_{j=1}^{A_{\calV,i}} B(\jobPair_i)^{(j)},
\]
where $\{B(\jobPair_i)^{(j)}\}$ are independent copies of $B(\jobPair_i)$ that are
independent of $(R_\calV,\vec A_\calV)$.
\end{definition}

The recursive definition of a busy period  corresponds exactly to the busy period tree representation introduced in \cref{sec:JJT-stability:trees}: the root corresponds to the initial work duration $R_\calV$, and each arrival during that duration initiates an independent descendant busy period whose distribution depends only on the class of its root job.

We now characterize the distribution of $B(\calV)$ via its Laplace-Stieltjes transform. This replaces the classical M/G/1 busy period identity, which relies on independence between service durations and arrivals and therefore does not apply in our setting.

\restatably\begin{theorem}[Busy Period Transform]\label{thm:busy-period-transform}
Consider $\theta \ge 0$, and let $\vec b$ be the vector defined by
\[
b_i = \lst{B(\jobPair_i)}(\theta).
\]
\*[subenv]
\* The busy period initiated by service-arrival joint distribution~$\calV$ has transform
\[
\lst{B(\calV)}(\theta) = \lst{\calV}\!\left(\theta, \vec b\right).
\]
\* The vector $\vec b$ is the least nonnegative solution to
\[
b_i = \jjt_i\!\left(\theta, \vec b\right), \qquad i=1,\ldots,n.
\]
\*/
\end{theorem}

As established in \cref{sec:JJT-stability:trees}, the set of jobs served during a
full busy period forms a multitype Galton-Watson branching process, where a node
of class~$i$ produces $A_{i,j}$ children of class~$j$. Standard branching-process
theory \citep[eq.~(II.4.1)]{harris_theory_1963} implies that the expected total number of jobs in a busy period is finite if and only if the spectral radius of the mean offspring matrix
$M = (\E{A_{i,j}})_{i,j}$ is strictly less than~$1$.

Fortunately in our case, the mean offspring matrix is rank one. This is because the system has time-homogeneous Poisson arrivals, meaning despite service-arrival dependence, we can show the following:

\restatably\begin{lemma}\label{lem:num-arrivals}
For all classes $i,j$, we have
\[
\E{A_{i,j}} = \lambda_j \E{R_i}.
\]
\end{lemma}

Thus $M$ is the outer product of two vectors. Specifically, $M$ is a rank-one matrix of the form
\[
M = \E{\vec R}\,\vec\lambda^{\top},
\]
meaning its unique nonzero eigenvalue is
\[
\rho = \E{\vec R}^{\top}\,\vec\lambda = \sum_{k=1}^n \lambda_k \E{R_k},
\]
the \emph{total effective load} defined in \cref{model:loads}. We therefore obtain
the usual stability condition.

\begin{theorem}[Stability Conditions]\label{thm:stability}
An $n$-class system with class~$k$ Poisson arrivals at rate $\lambda_k$ and class~$k$
jobs with job joint distribution $(R_k,\vec A_k)$ is stable (positive recurrent) if
$\rho < 1$.
\end{theorem}

Finally, we compute the expected length of busy periods. Remarkably, these expressions match those of the standard M/G/1 queue, with service times replaced by effective service
times.

\restatably\begin{theorem}\label{corr:busy-period-length}
For a service-arrival joint distribution $\calV$, the busy period initiated by
$\calV$ has mean
\[
\E{B(\calV)} = \E{R_\calV}
+ \sum_{j=1}^{n} \E{A_{\calV,j}} \E{B(\jobPair_j)}.
\]
Moreover, the busy period initiated by a class~$k$ job satisfies
\[
\E{B(\jobPair_k)} = \frac{\E{R_k}}{1-\rho}.
\]
\end{theorem}

\subsection{Class-restricted Busy Periods}\label{sec:JJT-stability:restricted}

We now define and analyze \emph{class-restricted busy periods}, which capture the
amount of work generated by arrivals of classes strictly lower than~$\ell$.
These objects will be central in our response-time analysis, since the response
time of a class~$\ell$ job depends only on arrivals of classes~$<\ell$, which would have higher priority.

Intuitively, a class~$<\ell$ busy period is obtained from a full busy period by discarding all arrivals of classes~$\ell,\ell+1,\ldots,n$ and all work generated
by those arrivals. Formally, this corresponds to modifying the branching construction so that only arrivals of classes~$<\ell$ produce descendants. 

\begin{figure}
    \centering
    \includegraphics[page = 2, width=\linewidth, trim= 0 0 0 0,clip]{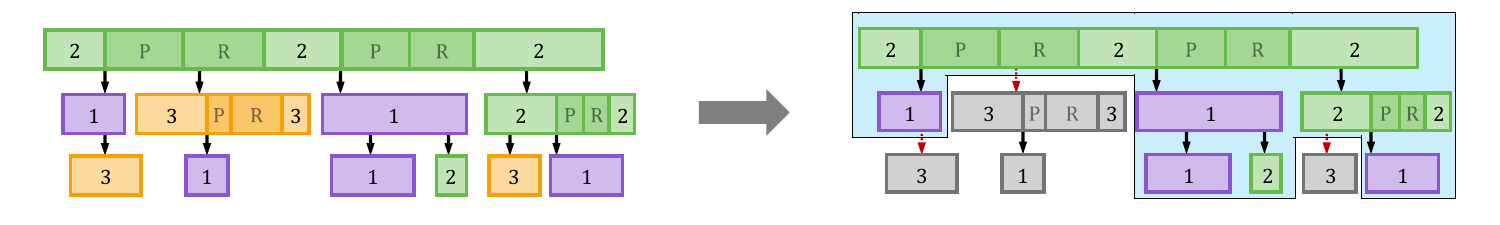}
    \caption{
    (Left)~A full busy period tree. Nodes are labeled with their class. Classes 1, 2, and 3 are purple, green, and orange respectively. Overhead is shaded darker. Pause and resumes are labeled with "P" and "R" respectively. Arrows are drawn from jobs to the jobs that arrived during their service.
    (Right)~Highlighting the class~$<3$ busy period tree with a blue background and black outline. Jobs of class~$3$ and their descendants are greyed out because they are not included in a class~$<3$ busy period tree.
    \label{fig:class-restricted}}
\end{figure}

\begin{definition}
The \emph{class~$<\ell$ busy period started by a service-arrival joint distribution
$\calV$}, denoted $B_{<\ell}(\calV)$, is the random variable satisfying
\[
B_{<\ell}(\calV) \stocheq
R_\calV + \sum_{i=1}^{\ell-1} \sum_{j=1}^{A_{\calV,i}} B_{<\ell}(\jobPair_i)^{(j)},
\]
where $\{B_{<\ell}(\jobPair_i)^{(j)}\}$ are independent copies of
$B_{<\ell}(\jobPair_i)$ and are independent of $(R_\calV,\vec A_\calV)$.
\end{definition}

This recursion corresponds to a busy period tree in which only nodes of class
$<\ell$ are allowed to generate descendants. Equivalently, it is obtained from the
full busy period tree by deleting all nodes of class~$\ge\ell$ and their
descendants, as seen in \cref{fig:class-restricted}.

As in the full busy period case, the distribution of a class~$<\ell$ busy period
is characterized by a fixed-point equation involving the job joint transforms.

\begin{corollary}\label{corr:class-busy-period-transform}
Consider $\theta \ge 0$ and $\ell \in \{1, \dots, n+1\}$, and let $\vec b$ be the vector defined by
\[
b_i = \begin{cases}
    \lst{B_{<\ell}(\jobPair_i)}(\theta) & \text{if } i < \ell \\
    1 & \text{if } i \ge \ell.
\end{cases}
\]
\* The class~$< \ell$ busy period initiated by service-arrival joint distribution~$\calV$ has transform
\[
\lst{B_{<\ell}(\calV)}(\theta) = \lst{\calV}\!\left(\theta, \vec b\right).
\]
\* The vector $\vec b$ is the least nonnegative solution to
\[
b_i = \begin{cases}
    \jjt_i\!\left(\theta, \vec b\right) & \text{if } i < \ell \\
    1 & \text{if } i \ge \ell.
\end{cases}
\]
\*/
\end{corollary}

Here the transform values corresponding to classes~$\geq\ell$ are set to~$1$ because we are essentially "zeroing out" the contributions to the busy period from arrivals of classes~$\geq \ell$ (and $e^0 = 1$).

The expected lengths of class~$<\ell$ busy periods follow from \cref{corr:busy-period-length}.
Recall that the quantity $\rho_{<\ell}$ is the effective load contributed by classes strictly less than~$\ell$.

\restatably\begin{corollary}\label{thm:class-busy-period-length}
For a service-arrival joint distribution $\calV$, the class~$<\ell$ busy period
initiated by $\calV$ has mean
\[
\E{B_{<\ell}(\calV)}
=
\E{R_\calV}
+
\sum_{j=1}^{\ell-1}
\E{A_{\calV,j}} \E{B_{<\ell}(\jobPair_j)}.
\]
Moreover, for a class~$k$ job,
\[
\E{B_{<\ell}(\jobPair_k)}
=
\frac{\E{R_k}}{1-\rho_{<\ell}},
\]
with $\rho_{<\ell}$ defined as in \cref{model:loads}.
\end{corollary}

Class~$<\ell$ busy periods will play a key role in the response-time analysis in
\cref{sec:response-time}, where they help describe the delay
experienced by a class~$\ell$ job due to higher-priority arrivals.

\subsection{When Do the Results of This Section Hold?}\label{sec:JJT-stability:assumptions}\label{sec:local-arrivals}

Our job joint transform framework holds for overhead models and queueing systems beyond just the one we study in this work. In order for our analysis from this section to hold for a system, we need two "ingredients". First, the general branching-process approach requires that we can model the arrival process as being "local" to each job. This locality means that the set of jobs served during a busy period can be modeled as a multitype Galton-Watson process, where each node's offspring distribution depends only on its class and remains independent of the "global" timeline. Second, to derive the formulas and expressions for stability and busy period lengths \see\cref{thm:stability, corr:busy-period-length}, the system must satisfy the identity $\E{A_{k,i}} = \lambda_i\E{R_k}$ for all classes $i$ and~$k$.  This property makes the mean offspring matrix $M$ a rank-one matrix, which is the primary reason the resulting expressions for stability and mean busy period length are intuitive and computationally simple.

These two ingredients are naturally provided by time-homogeneous Poisson arrivals with "arrival-sensitive" service. Other systems like this include both preempt-repeat and preeempt-restart-different, along with many other preemption overhead models. Our framework extends even beyond these systems as well; for instance, discrete-time models with slotted arrivals and service (that is, i.i.d. arrivals each time step) and i.i.d. batch-Poisson processes also "contain" both ingredients.

Even in systems that do not possess the rank-one property, such as those with arrival rates that depend on the class of job in service, the first ingredient, "local arrivals", still holds. This allows the multitype branching structure to remain a valid tool for stability analysis, meaning the busy period transform results still hold, even if the final formulas for stability and expected lengths may become more complex.

We analyze the job joint transform for a variety of the other systems mentioned here in \cref{sec:extensions}.

\section{Deriving the Joint Job Transform for Our Overhead Model}\label{sec:jjt-derivation}
Now that we have developed general stability conditions for systems with service-arrival dependency in terms of the job joint transform, we can use the specific overhead dynamics of our model \see\cref{model} to derive the exact form of the job joint transform for our system. By deriving the job joint transform for our system and the load of the system, we will have the necessary pieces to complete the busy period and stability analysis as in \cref{thm:stability} for our model.

\subsection{Overhead Chains}
Recall from \cref{model:overhead} that the arrival of class~$<k$ jobs during the service of a class~$k$ job triggers overhead and thus increases the effective service time of the class~$k$ job.
But by how much does each preempting arrival extend our job's service time?
Each time a class~$k$ job is interrupted, a "chain" of overheads occurs before the job continues service towards its original size, consisting of some number of pause-resume pairs.
The chain continues whenever a resume fails, and it ends when a resume succeeds \see\cref{alg:pprio_with_overhead}. See \cref{fig:preemption-chain} for an illustration.

\begin{definition}\label{def:preemption-chain}
    A \emph{class~$k$ overhead chain} is the sequence of pause and resume overheads between a class~$k$ job being preempted and continuing service towards its original size.
    \*[(a)]A \emph{link of an overhead chain} is the pair made up of a pause overhead and the resume overhead directly succeeding it.
    \* an overhead chain forms another link if the resume of the current link fails (i.e., an arrival of class lower than the preempted job arrives during the resume overhead of the current link) \see\cref{alg:pprio_with_overhead}.
    \* The \emph{final link of an overhead chain} occurs when the resume of the current link succeeds (i.e., no preempting arrivals occur during the resume overhead of the current link).
    \*/
\end{definition}

Thus the effective service time of a class~$k$ job is made up of the original size~$S_k$ plus, for each class~$< k$ arrival during $S_k$, the length of one \emph{overhead chain}.

\cref{lem:JJT-complete} gives the job joint transform for class~$k$ jobs in terms of $\lst\calO_k(\theta, \vec{z})$, the service-arrival joint transform of a single overhead chain. Thus, deriving $\lst\calO_k$ is all that remains to derive the job joint transform for a class~$k$ job.


\begin{figure}
    \centering
    {\includegraphics[page = 3, width=\linewidth, trim= 100 25 100 0,clip]{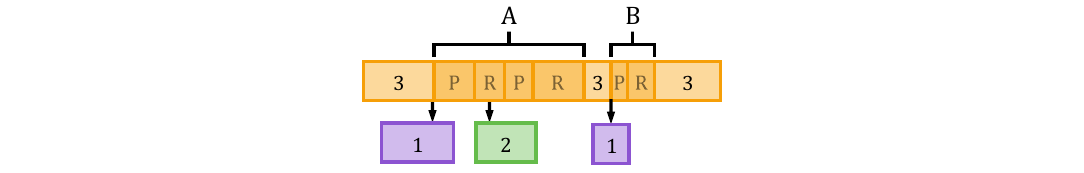}}
    \caption{
    A class~$3$ job and the arrivals during its service. The shaded orange region marked A is a class~$3$ overhead chain made up of two links, and the other shaded orange region marked B is another class~$3$ overhead chain made up of only one link.}
    \label{fig:preemption-chain}
\end{figure}

In order to derive the service-arrival joint transform of an overhead chain, let us begin by reasoning through how the links of an overhead chain are formed. The first link occurs when the job was initially preempted. The next link forms if a class~$< k$ job arrives during the resume portion of the current link. The overhead chain has its final link when a resume completes without any arrivals from class~$< k$ jobs occurring. This leads us to the following observation.

\begin{observation}\label{remark:links}
    By a standard property of exponential random variables and Laplace-Stieltjes transforms, it follows that the links in an overhead chain form geometrically with probability $\lst{D_k}(\lambda_{< k})$ of being the final link in a chain. This means the number of future links or arrivals in an overhead chain is independent of the number of links or arrivals in the chain that have already occurred.
\end{observation}

Characterizing the exact distribution of the length of each overhead chain requires understanding the joint distribution between the length of the overhead chain and the number of preempting arrivals during that overhead chain. Thus we must derive the service-arrival joint transform $\lst\calO_k(\theta, \vec{z})$ of the duration of a class~$k$ overhead chain.

\begin{lemma}\label{lem:JJT-chain}
    The service-arrival joint transformation for a single class~$k$ overhead chain is:
    \[
    \lst\calO_k(\theta, \vec{z})
        = \frac{\lst{C_k}\gp*{
            \theta + \lambda(1 - z)} \, \lst{D_k}\gp*{
            \theta + \lambda_{\geq k}(1 - z_{\geq k}) + \lambda_{< k}}}
            {1 - \lst{C_k}\gp*{
            \theta + \lambda(1 - z)}  \gp[\big]{
                \lst{D_k}\gp*{
            \theta + \lambda(1 - z)} - \lst{D_k}\gp*{
            \theta + \lambda_{\geq k}(1 - z_{\geq k}) + \lambda_{< k}}
            }}.
    \]
\end{lemma}
\begin{proof}
    We complete this proof through the lens of the method of collective marks \see\cref{sec:MOCM}. The joint transform $\lst\calO_k(\theta, \vec{z})$ can be interpreted as the probability that the entire overhead chain is unmarked by both Poisson arrivals at rate $\theta$ during its duration, as well as by job arrivals of class~$i$ marked with probability $z_i$.
    Define the events
    \[
    CU &= \text{chain unmarked},\\
    LU &= \text{first link unmarked},\\
    PA &= \text{preempting arrival during first resume}
    \]
    Then, using \cref{remark:links}, it follows that
    \[
    \lst\calO_k(\theta, \vec{z})
    &
    = \P{CU}
    = \P{CU \cap PA} + \P{CU \cap PA^C}
    = \P{CU \cap PA} + \P{LU \cap PA^C}
    \\ &
    = \P{LU \cap PA}\P{CU} + \P{LU \cap PA^C}
    \\ &
    = (\P{LU} - \P{LU \cap PA^C})\P{CU} + \P{LU \cap PA^C}
    \\
    \Rightarrow \P{CU} &= \frac{\P{LU \cap PA^C}}{1 - (\P{LU} - \P{LU \cap PA^C})}.
    \]
    It remains to calculate the probabilities $\P{LU}$ and $\P{LU \cap PA^C}$. $\P{LU}$ is the probability that the first link of an overhead chain is unmarked. This is equal to the probability that the pause overhead is unmarked \emph{and} the resume overhead is unmarked (which are independent), which simply means:
    \[
    \P{LU} &= \lst{C_k}(\theta + \lambda(1 - z))\lst{D_k}(\theta + \lambda(1 - z)).
    \]
    $\P{LU \cap PA^C}$ is the probability that the first link of the overhead chain is unmarked \emph{and} there are no further links in the overhead chain. This is equivalent to the probability that the pause overhead is unmarked, the resume overhead is unmarked, \emph{and} the resume overhead has no class~$< k$ arrivals (occurring at rate $\lambda_{< k}$) occur \emph{at all}. This means that
    \[
    \P{LU \cap PA^C} &= \lst{C_k}(\theta + \lambda(1 - z))\lst{D_k}(\theta + \lambda_{\geq k}(1 - z_{\geq k}) + \lambda_{< k}).
    \]
    Plugging these to the expression we found for $\lst\calO_k(\theta, \vec{z})$ yields the desired expression.
\end{proof}

\subsection{From the Overhead Chain Joint Transform to the Job Joint Transform}

Now that we know the expression for the duration of a single overhead chain, and we know that a job's effective service time is made up of its original size plus one overhead chain for each preemption during its original size, we obtain the expression for the job joint transform of a class~$k$ job in a preemptive priority system with our model \see\cref{model:overhead} of preemption overhead.

\begin{theorem}\label{lem:JJT-complete}
    The job joint transform for a class~$k$ is
    \[
        \jjt_k(\theta, \vec{z})
        = \lst{S_k}\gp*{
            \theta + \lambda_{\geq k}(1 - z_{\geq k}) + \lambda_{< k}(1 - z_{< k}\lst\calO_k(\theta, \vec{z}))}
        ,
    \]
    where $\lst\calO_k(\theta, \vec{z})$ is as in \cref{lem:JJT-chain}.
\end{theorem}

\begin{proof}
    Consider a class~$k$ job. Via the method of collective marks, the mark rate during the effective service time is equal to the base mark rate $\theta$ plus the rate of marked nonpreemptive arrivals $\gp*{\sum_{i=k}^n\lambda_i(1 - z_i)} = \lambda_{\geq k}(1 - z_{\geq k})$, plus the rate of preemptive arrivals that are either marked \emph{or} cause a marked overhead chain $\sum_{i=1}^{k - 1} \lambda_{i}(1 - z_i\lst\calO_k(\theta, \vec{z})) = \lambda_{< k}(1 - z_{< k}\lst\calO_k(\theta, \vec{z}))$. Recall that $\calO_k(\theta, \vec{z})$ denotes the probability that a class~$k$ overhead chain is unmarked, which we found in \cref{lem:JJT-chain}.
\end{proof}

The final step of defining our system as an service-arrival-dependent system is by confirming we can compute stability conditions from the job joint transform. In order to do so, we must be able to actually derive values for the class~$k$ pause and resume loads $\gamma_k \coloneqq \lambda_k\E{C^*_k}$ and $\delta_k \coloneqq \lambda_k\E{D^*_k}$, which we can do via computing the expectation of $C^*_k$ and $D^*_k$.

\begin{lemma}\label{lem:overhead-load}
The expected time spent on pause and resume overhead per class~$k$ job are
    \[
    \E{C^*_k} &= \frac{\lambda_{<k}\E{S_k} \E{C_k}}{\lst{D_k}(\lambda_{<k})}, &
    \E{D^*_k} &= \frac{\lambda_{<k}\E{S_k} \E{D_k}}{\lst{D_k}(\lambda_{<k})}.
    \]

\end{lemma}
\begin{proof}
    Let us begin by computing $\E{C^*_k}$, the expected amount of time spent on pause overhead per class~$k$ job. This is just the expected number of overhead chains per class~$k$ job, times the expected number of pauses (or links) per class~$k$ overhead chain, times the expected length of one pause overhead, by Wald's equation. The average number of overhead chains per class~$k$ job is simply the average number of preempting arrivals during a class~$k$ job's original service length, which is just $\lambda_{<k}\E{S_k}.$
    
    The links in an overhead chain are formed geometrically (\cref{remark:links}), as the next pause-resume pair occurs only if a class~$k$ arrival occurs during the resume overhead of the current pause-resume pair. Thus the average number of links in an overhead chain is the reciprocal of the probability that an overhead chain terminates. The probability that an overhead chain terminates is the probability that no preempting arrivals occur during a resume overhead, which is $\lst{D_k}(\lambda_{<k})$.  Multiplying the values we acquired gives us the desired formula. The proof for $\E{D^*_k}$ is the same, except we replace $C_k$ with $D_k$.
\end{proof}

\section{Computing Response Time}\label{sec:response-time}

Now that we have completed busy period and stability analysis for our system, it remains to use these values to actually find the response time distribution of class~$k$ jobs under preemptive priority with overhead. Our main result is stated below.

\begin{theorem}\label{thm:response-time}
    The Laplace-Stieltjes transform of response time for a class~$k$ job under a preemptive priority scheduling policy with preemption overhead model as in \cref{model:overhead} is
    \[
    \lst{T_k}(\theta) = \frac{(1-\rho_{<k}-\rho_k)\lst{B_{<k}(\calJ_k)}(\theta)}{1-\rho_{<k}-\rho_k \gp[\big]{\lst{B_{<k}(\calJ_k)}}_\e(\theta)}\widetilde{X_k^*}(\theta),
    \]
    where
$X_k^*$ is as in \cref{lem:x-distribution}.
\end{theorem}

To prove \cref{thm:response-time}, we proceed as follows.
We first show that the response time of a class~$k$ job under preemptive priority with overhead is equivalent to the response time of a job in an M/G/1/setup system \see\cref{def:setup} with certain parameters.
This means that the response time of a class~$k$ job decomposes into the independent sum of a steady-state amount of work in system and another quantity related to overhead values and class~$<k$ work in system, which we call $X_k^*$.
We then characterize the random variable $X_k^*$ and derive its transform $\widetilde{X_k^*}$, which, together with \cref{thm:supersystem_equivalence} and \cref{thm:busy-period-transform}, yields the desired expression for $\lst{T_k}(\theta)$.

\subsection{Segment Trees}\label{sec:tree:response-time}

We consider the response time $T_k$ of an arbitrary class~$k$ job. We shall work through a series of observations about busy periods that allow us to view $T_k$ in our preemptive priority system with overhead as the response time distribution of a job in the M/G/1/setup (\cref{def:setup}) with carefully chosen system parameters.

To better visualize busy periods and understand the observations we make about busy periods, we  will work with a notion similar to that of a \emph{busy period tree} \see\cref{sec:JJT-stability:trees} as a way of encoding the arrival and service history of jobs during a busy period.
In the stability analysis, it was sufficient to view the busy period tree at the level of jobs and arrivals.
For response time analysis, however, we require a finer representation that allows us to understand the exact order in which service is performed.
Accordingly, in this section we view the busy period on a "segment level", as a \emph{segment tree}.

\begin{definition}\label{def:segment-tree}
A \emph{segment} is a maximal contiguous interval during which a single job is served.
Because jobs may be preempted and later resumed, the service of a job may consist of multiple segments, which we order from earliest to latest (e.g., a job's first segment starts when it first enters service).

A \emph{segment tree} is a directed, rooted tree that encodes the arrival and service history during a single full busy period.
The nodes of the tree are segments.
There are two types of edges between segments:
\begin{itemize}
    \item a \emph{continuation edge} connects two consecutive segments of the same job;
    \item an \emph{arrival edge} connects a segment to the first segment of a job that arrives during
    its service.
\end{itemize}
\Cref{fig:segment-tree} illustrates how a timeline of arrivals and service corresponds to the segment tree.
\end{definition}

\begin{figure}
    \centering
    \includegraphics[page = 4, width=\linewidth, trim= 10 0 10 0,clip]{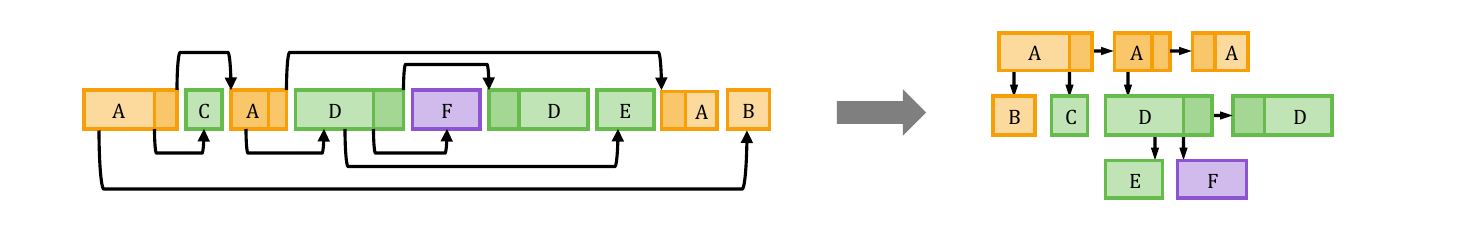}
    \caption{Two representations of the arrival sequence and service order of jobs in a busy period. Job F is of class~$1$ (periwinkle); jobs C, D, and E are of class~$2$ (orange); and jobs A and B are of class~$3$ (green). Each block is a segment of a job. Original work is colored a lighter shade and overhead work is colored a darker shade. 
            (Left)~Timeline of the service order of segments in the system. Time goes from left to right. Continuation edges are shown above the timeline. Arrival edges are shown below the timeline.
    (Right)~Segment tree representation of the same arrival sequence. Horizontal edges are continuation edges. Vertical edges are arrival edges.}
    \label{fig:segment-tree}
\end{figure}

When the tree contains all segments served during a single full busy period (i.e., a maximal interval during which the server is continuously busy), we refer to it as a \emph{full segment tree}, or simply a \emph{full tree}.

Our goal is to understand the response time of a fixed class~$k$ job.
To do so, we first identify which parts of the full segment tree can possibly influence the service experienced by a class~$k$ job.
Because the scheduling discipline is preemptive priority, not all segments in the full tree are relevant to this job’s response time.

\begin{observation}\label{obs:class-k-restriction}
Before any class~$>k$ segment can begin service, there must be no class~$\leq k$ work remaining in the
system, as they would receive priority over the class~$> k$ job.
\end{observation}

Motivated by \cref{obs:class-k-restriction}, suppose that we remove all continuation and arrival edges into class~$> k$ segments and consider all resulting connected components. Then each class~$k$ job would not be impacted by any segments outside its connected component, as the class~$k$ job definitionally must arrive and depart during service of this connected component. We will call the connected components made by removing all arrival and continuation edges into class~$> k$ nodes \emph{class~$k$ relevant subtrees}. See \cref{fig:relevant-subtree} for an illustration of the class~$2$ relevant subtrees resulting from cutting all arrival and continuation edges into class~$> 2$ jobs in a full segment tree.

\begin{figure}
    \centering
    \includegraphics[page = 5, width=\linewidth, trim= 10 0 10 3,clip]{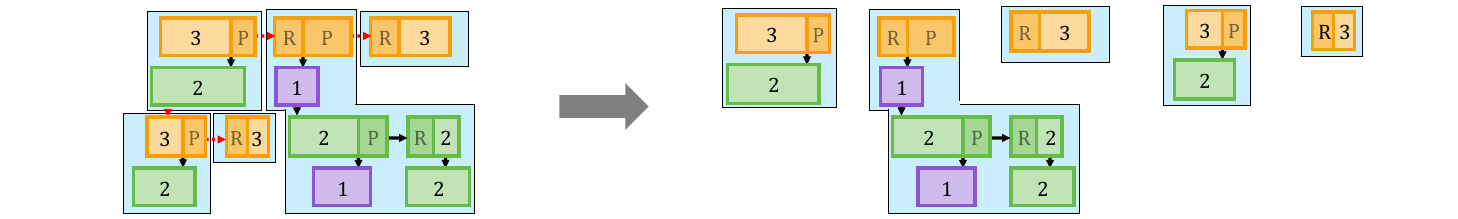}
    \caption{(Left)~Class~$2$ relevant subtrees of a segment tree highlighted by blue backgrounds. Segments are labeled with their class. Red dotted arrows indicate arrival or continuation edge into class~$> 2$ segment. (Right)~Separation of the class~$2$ relevant subtrees. The set of segments in each subtree is served contiguously, and notably, any class~$\leq 2$ job arrives and departs entirely during its relevant subtree.}
    \label{fig:relevant-subtree}
\end{figure}

\begin{observation}\label{obs:restriction}
The response time of a class~$k$ job depends only on segments contained in its class~$k$ relevant
subtree.
\end{observation}

\subsection{Superjobs and Supersetups}

Now that we know analyzing a class~$k$ job's response time in the full tree is equivalent to analyzing
its response time inside its class~$k$ relevant subtree, it remains to understand the order in which segments in a class~$k$ relevant subtree are
served.

\begin{observation}\label{obs:fcfs}
Within a class~$k$ relevant subtree, a class~$k$ job cannot begin service until the previous class~$k$
job (if any) has completed service, since jobs are served within each class in FCFS order.
\end{observation}

Motivated by \cref{obs:fcfs}, we construct a decomposition of a class~$k$ relevant subtree that makes this FCFS structure explicit.
Specifically, we remove all arrival edges into class~$k$ jobs.
The removal of these edges separates the class~$k$ relevant subtree into two types of connected components.
We introduce the following terminology to refer to these components:
\* a connected component that contains a class~$k$ segment is called a \emph{class~$k$ superjob},
\* a connected component that contains no class~$k$ segments is called a
\emph{class~$k$ supersetup}.
\*/
See \cref{fig:super-system} for an illustration of the resulting superjobs and supersetups.
One can show that no other types of connected components emerge from these cuts (e.g. no component shares segments from two different class~$k$ jobs, because the arrival edge going into the second class~$k$ job would have been removed).
We next explain why these class~$k$ superjobs and supersetups are the natural units for response-time analysis.

\begin{figure}
    \centering
    \includegraphics[page =6, width=\linewidth, trim= 20 7 20 3,clip]{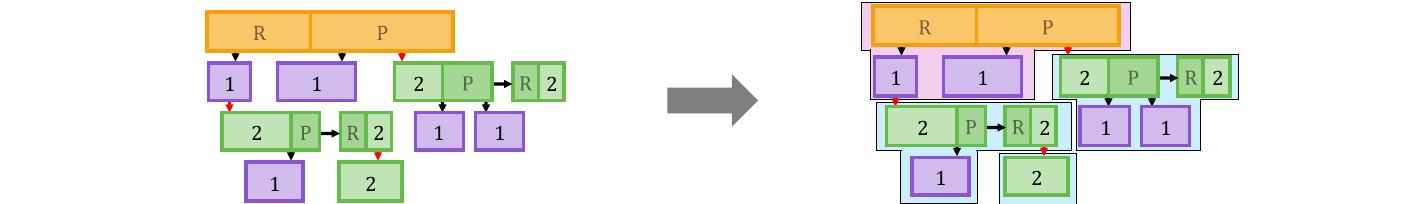}
    \caption{(Left)~A single class~$2$ relevant subtree. Segments are labeled with their class. Red dotted arrows denote arrival edges into class~$2$ jobs.
    (Right)~Partition of the segments into class~$2$ superjobs and a class~$2$ supersetup. Shaded areas denote boundaries between different class~$2$ superjobs or supersetups. The pink shaded area denotes a supersetup, whereas each blue shaded area denotes a separate superjob. The superjobs are served in FCFS order according to the arrival time of their respective first segments.}
    \label{fig:super-system}
\end{figure}

A class~$k$ superjob is a connected component that contains a class~$k$ segment. Because these components result from cutting all \emph{arrival} edges into class~$k$ jobs in a class~$k$ relevant subtree, each superjob must contain \emph{all} segments of exactly one class~$k$ job. That is, each class~$k$ superjob corresponds to exactly one class~$k$ job. 
Moreover, by the same reasoning that motivated the relevant subtree definition, a class~$k$ job cannot finish service until all class~$< k$ jobs that arrive during its execution, and any lower-class jobs that arrive during those, and so on, have also completed, which gives us the following observation.

\begin{observation}\label{obs:superjob-size}
    \*[subenv]
    \* The \emph{size} of a class~$k$ superjob---that is, the total amount of work of all jobs contained in the superjob---is distributed as $B_{<k}(\jobPair_k)$ \see\cref{def:class-busy-period}, the length of a class~$< k$ busy period started by a class~$k$ job.
    \* Every class~$k$ superjob is served in a contiguous interval that starts when the class~$k$ job begins service and ends when it completes.\footnote{In the terminology of the M/G/1 scheduling literature, e.g. \citet{harchol-balter_performance_2013}, the size of the class~$k$ superjob started by a class~$k$ job is the \emph{residence time}.}
    \*/
\end{observation}

A class~$k$ supersetup is a connected component that does not contain any class~$k$ segments.

\begin{observation}\label{obs:superjob-order}
    \*[subenv]
    \* Within a class~$k$ relevant subtree, at most one supersetup exists, and if it does, it is at the root of that subtree.
    \* Within a class~$k$ relevant subtree, the supersetup (if it exists) is served first, after which
class~$k$ superjobs are served contiguously and in FCFS order of their class~$k$ root segments.
    \*/
\end{observation}

This observation tells us that, within a class~$k$ relevant subtree, all work
that is not part of a class~$k$ superjob is confined to a single initial block of service.
Thus, from the perspective of class~$k$ jobs, the system consists of an initial amount of
``extra work'' (the supersetup, if present), followed by a sequence of class~$k$ superjobs that
must be completed one at a time.
This structure suggests viewing the system as an FCFS queue with additional work incurred at the
start of a busy period, which motivates the equivalence to an M/G/1/setup system.

Thus, the way we can view the preemptive priority system with overhead through the eyes of a class~$k$ job is as one in which class~$k$ superjobs and supersetups are the primitive entities, which are nonpreemptively served in FCFS order. We call this the \emph{supersystem}.

\subsection{Equivalency to M/G/1/setup}

By reasoning about service order in this manner, we see that a class~$k$ superjob arrives at the same time as its corresponding class~$k$ job, and departs when its last segment completes. But the last segment of the class~$k$ superjob to complete will always be the last segment of the corresponding class~$k$ job. Thus the arrival and departure of class~$k$ jobs in the original system exactly correspond to the arrival and departure of class~$k$ superjobs in the supersystem, and so they have the same response time distribution. However, the supersystem is exactly like an ordinary M/G/1 with additional work due to supersetups. That is, the supersystem happens to be an M/G/1/setup.

\begin{definition}\label{def:setup}
    \*[subenv, afterheading]
    \* The \emph{M/G/1/setup} with arrival rate~$\lambda^*$, size distribution~$S^*$, and \emph{setup time} distribution~$U^*$ is like a standard M/G/1 with FCFS service, except whenever a job arrives to an empty system, a delay, called a \emph{setup time}, of length sampled from $U^*$ occurs before any job begins service. We denote the response time distribution in the M/G/1/setup by~$T^*$.
    \* The \emph{extra work} distribution, denoted~$X^*$, is the remaining setup time length as seen by a job that either arrives to an idle system or during a setup time. It is a standard result that $X^*$ is a mixture between $U^*$ and its excess $(U^*)_\e$:
    \[
        X^* \stocheq \begin{cases}
            U^* & \text{with probability } \frac{1}{1 + \lambda^* \E{U^*}} \\
            (U^*)_\e & \text{with probability } \frac{\lambda^* \E{U^*}}{1 + \lambda^* \E{U^*}}.
        \end{cases}
    \]

    \* Classic results from \citet{fuhrmann_stochastic_1985} tell us that the response time $T^*$ of an arriving job can be decomposed into the independent sum
    \[
    T^* \stocheq W^* + S^* + X^*,
    \]
    where $W^*$ is the steady-state workload in the M/G/1 with arrival rate $\lambda^*$ and size distribution $S^*$ (excluding setup time),
    $S^*$ is the job’s own service requirement, and $X^*$ is the extra work, as defined above.
    \*/
\end{definition}

Then analyzing the response time of class~$k$ jobs boils down to analyzing the response time of class~$k$ superjobs in their relevant class~$k$ subtrees, which is an M/G/1/setup system. In order to state our result, we need the following definition:

\begin{definition}\label{def:early}
    An \emph{early class~$k$ job} is any class~$k$ job that arrives when no other class~$k$ jobs in the system have received service so far.
\end{definition}

An early class~$k$ job as viewed through the supersystem lens is a class~$k$ superjob that arrives to either an idle system or to a supersetup. In \cref{fig:super-system}, there are two early jobs.

With \cref{def:early} in hand, we can state the key result of this section.
It reduces  proving \cref{thm:response-time} to characterizing the distribution of~$X_k^*$, which we do in \cref{sec:extra-work}.

\begin{definition}\label{def:extra_work_class_k}
    We denote by $X_k^*$ the distribution of the time between the arrival of an early class~$k$ job and the next moment when some class~$k$ job (possibly a different one) begins service.
\end{definition}

\begin{theorem}\label{thm:supersystem_equivalence}
    The response time distribution~$T_k$ of class~$k$ jobs in the preemptive priority system with overhead is equivalent to the response time distribution $T^*$ in an M/G/1/setup, where:
    \begin{itemize}
        \item The arrival rate is $\lambda^* = \lambda_k$.
        \item The size distribution is $S^* \stocheq B_{<k}(\jobPair_k)$ \see\cref{obs:superjob-size}.
        \item The distribution of extra work is $X^* \stocheq X^*_k$ \see\cref{def:extra_work_class_k}.
    \end{itemize}
\end{theorem}

\begin{proof}
We proceed in two steps.
First, we show that from the perspective of a class~$k$ job, the preemptive priority system with
overhead can be viewed as an instance of an M/G/1/setup system.
Second, we identify the parameters of this M/G/1/setup.

Consider the busy period tree representation of the preemptive priority system with overhead.
By \cref{obs:class-k-restriction,obs:restriction}, the response time of a class~$k$ job depends only
on segments contained in its class~$k$ relevant subtree.
Hence, it suffices to restrict attention to a fixed class~$k$ relevant subtree.

Within a class~$k$ relevant subtree, \cref{obs:fcfs} implies that class~$k$ jobs are served in FCFS
order.
Removing all arrival edges into class~$k$ jobs partitions the subtree into connected components that
are either class~$k$ superjobs or a class~$k$ supersetup.
Each class~$k$ superjob contains all segments of exactly one class~$k$ job together with all
recursively generated class~$<k$ work, and each connected component is served contiguously \see\cref{obs:superjob-size}.
Moreover, by \cref{obs:superjob-order}, the supersetup (if it exists) is served first, after which
class~$k$ superjobs are served contiguously and in FCFS order of their class~$k$ root segments.

Therefore, from the viewpoint of class~$k$ jobs, the system can be viewed as one in which the
primitive entities are the class~$k$ superjobs, served nonpreemptively in FCFS order, together with
a possible initial block of extra work (the supersetup) occurring at the start of a busy period.
We refer to this system as the \emph{supersystem}.
The supersetup corresponds to a setup time, and each class~$k$ superjob corresponds to a job served
in FCFS order in the supersystem.

A class~$k$ superjob arrives at exactly the same time as its corresponding class~$k$ job,
and departs when the last segment of that class~$k$ job completes service.
In particular, the arrival time and departure time of each class~$k$ job coincide with those of its
associated class~$k$ superjob.
It follows that class~$k$ jobs and class~$k$ superjobs have identical response time distributions.
Thus, the response time distribution $T_k$ is equal in distribution to the response time $T^*$ in an
M/G/1/setup system.

We now identify the parameters of the equivalent M/G/1/setup.
Class~$k$ jobs arrive according to a Poisson process of rate $\lambda_k$.
Because of the one-to-one correspondence between class~$k$ jobs and class~$k$ superjobs, the arrival
process of superjobs is also Poisson with rate $\lambda^* = \lambda_k$.

By construction of superjobs, the total work from the instant a class~$k$ job first begins service until it finally departs consists of that job’s effective service time together with all recursively generated class~$<k$ work that arrives during that service time.
Equivalently, each class~$k$ superjob is a class~$<k$ busy period initiated by a class~$k$ job, and hence
\[
S^* \stocheq B_{<k}(\jobPair_k).
\]

Finally, let $U$ denote the duration of the supersetup in a class~$k$ relevant subtree, with $U=0$ if no supersetup is present.
If a class~$k$ job arrives before any class~$k$ service has begun in that subtree (an early arrival), then the time until the next class~$k$ service start is exactly the residual portion of $U$ present at the time of arrival.
By definition of the M/G/1/setup \cref{def:setup}, this quantity has the same distribution as the extra-work random variable $X^*$.
Thus,
\[
X^* \stocheq \text{time from an early class~$k$ arrival to the next class~$k$ service start} \eqqcolon X_k^*.
\]
Having identified the parameters $(\lambda^*, S^*, X^*)$ as stated, the proof is complete.
\end{proof}

\subsection{Computing Distribution of Extra Work}
\label{sec:extra-work}

We showed in \cref{sec:tree:response-time} that computing the class~$k$ response time reduces to determining the distribution of $X_k^*$. 
An early class~$k$ job can arrive while the server is in one of several possible states, illustrated in \cref{fig:remaining-supersetup}.
The next lemma enumerates these cases.

\restatably\begin{lemma}\label{lem:x-distribution}
    The distribution $X_k^*$ for an early class~$k$ arrival to a preemptive priority system with overhead is:
    \renewcommand{\frac}{\tfrac}
    \[
    \allowdisplaybreaks
    \label{eq:x-distribution:a}
    &\gp[\big]{B_{<k}(\jobPair_i)}_\e
    &\text{w.p.}\quad
    &(1 - \rho) \cdot \gp[\Big]{\frac{\rho_i}{1 - \rho_{\leq k}}}
    && \text{for $i < k$,}
    &
    \\*
    \label{eq:x-distribution:b}
    &\gp[\big]{B_{<k}(C_j + \jobPair_i)}_\e
    &\text{w.p.}\quad
    & \lambda_i\sigma_j \cdot \gp[\Big]{\frac{\E{C_j + R_i}}{1 - \rho_{\leq k}}}
    && \text{for $i < k < j$,} \\
    \label{eq:x-distribution:c}
    &\gp[\big]{B_{<k}(C_j)}_\e
    &\text{w.p.}\quad
    & \lambda_\ell\sigma_j \cdot \frac{\E{C_j}}{1 - \rho_{\leq k}}
    && \text{for $k \leq \ell < j$,} \\
    \label{eq:x-distribution:d}
    &\gp[\big]{B_{<k}(D_j^{(1)} + C_j + \jobPair_i)}_\e
    &\text{w.p.}\quad
    &\lambda_i\sigma_j \cdot \gp[\Big]{\frac{1 - \lst{D_j}(\lambda_{<j})}{\lst{D_j}(\lambda_{<j})}}\cdot \gp[\Big]{\frac{\E{D_j^{(1)} + C_j + R_i}}{1 - \rho_{ \leq k}}}
    && \text{for $i < k < j$,} \\
    \label{eq:x-distribution:e}
    &\gp[\big]{B_{<k}(D_j^{(1)} + C_j)}_\e
    &\text{w.p.}\quad
    &\lambda_\ell\sigma_j \cdot \gp[\Big]{\frac{1 - \lst{D_j}(\lambda_{<j})}{\lst{D_j}(\lambda_{<j})}}\cdot \gp[\Big]{\frac{\E{D_j^{(1)} + C_j}}{1 - \rho_{ \leq k}}} 
    && \text{for $k \leq \ell < j$,} \\
    \label{eq:x-distribution:f}
    &B_{<k}(C_j)
    &\text{w.p.}\quad
    & \sigma_j \cdot \gp[\Big]{\frac{1 - \rho_{< k}}{1 - \rho_{\leq k}}}
    && \text{for $k < j$,} \\
    \label{eq:x-distribution:g}
    &(D_j^{(0)})_\e + B_{<k}\gp{C_j}
    &\text{w.p.}\quad
    & \sigma_j
    \cdot \gp[\Big]{\frac{1 - \lst{D_j}(\lambda_{<j})}{\lst{D_j}(\lambda_{<j})}}
    \cdot \gp[\Big]{\frac{1 - \rho_{< k}}{1 - \rho_{\leq k}}}
    && \text{for $k < j$,} \\*
    \label{eq:x-distribution:h}
    &0
    &\text{w.p.}\quad
    &(1 - \rho)\cdot \gp[\Big]{\frac{1 - \rho_{< k}}{1 - \rho_{\leq k}}},
    \]
    where $D_j^{(0)}$ and $D_j^{(1)}$ are defined in \cref{lem:resume-tfms}.
\end{lemma}

\begin{figure}
    \centering
    \includegraphics[page = 7, width=\linewidth, trim= 0 0 0 0,clip]{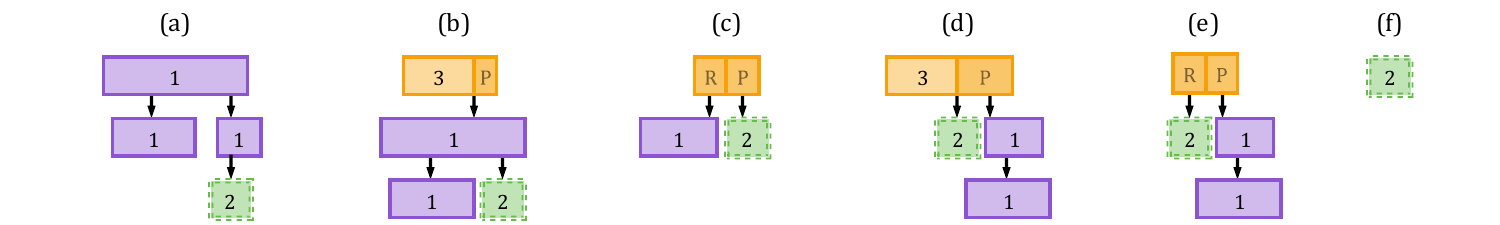}
    \caption{Pictured above are the various states of the system to which an early class~$2$ job (green, with dotted border) could arrive. From left to right, an early class~$2$ job could (a) arrive during the busy period started by a lower-class job that arrived to an idle system, (b) arrive during the busy period started by the pause overhead triggered by an earlier class~$\leq 2$ job that arrived during the service of a class~$> 2$ job, (c) arrive during the busy period started by the pause overhead triggered by an earlier class~$\leq 2$ job that arrived during the resume of a class~$> 2$ job, (d) arrive during the service of a class~$> 2$ job, triggering pause overhead, (e) be the first higher-priority job to arrive during the resume of a class~$> 2$ job, triggering pause overhead, or (f) arrive to an idle system.}
    \label{fig:remaining-supersetup}
\end{figure}

\begin{proof}
    See \cref{app:x-proofs} for proof.
\end{proof}

\restatably\begin{lemma}\label{lem:resume-tfms} We denote by $D_j^{(0)}$ the distribution of the length of a class~$j$ resume conditional on it having succeeded, and similarly for $D_j^{(1)}$ conditional on having failed.
The transforms of the conditional class~$j$ resume distributions are:
\[
\lst D_j^{(0)}(\theta) &= \frac{\lst{D_j}(\theta + \lambda_{< j})}{\lst{D_j}(\lambda_{< j})}
,
&
\lst D_j^{(1)}(\theta) &= \frac{\lst{D_j}(\theta) - \lst{D_j}(\theta + \lambda_{< j})}{1 - \lst{D_j}(\lambda_{< j})}
.
\]
\end{lemma}

\begin{proof}
These are standard Laplace transform exercises. See \cref{app:x-proofs} for proof.
\end{proof}

\subsection{Proof of Main Result}
Combining \cref{thm:supersystem_equivalence, lem:x-distribution, lem:resume-tfms} yields \cref{thm:response-time}. 
The results of \cref{lem:x-distribution,thm:busy-period-transform} are sufficient to compute $\widetilde{X_k^*}$ because:
(1)~the terms in \cref{lem:x-distribution} are independent;
(2)~the transform of a sum of independent variables equals the product of their transforms;
(3)~the transform of a mixture is the weighted sum of the component transforms; and
(4)~the transform of the excess of a variable $V$ can be expressed in terms of $\lst{V}$, as noted in \cref{def:excess}.

\section{Extensions}\label{sec:extensions}

Below, we derive the job joint transform \see\cref{def:jjt} for alternative overhead models and for other systems with service-arrival dependency. The stability results from \cref{sec:JJT-stability} apply to all the below systems with their respective job joint transforms. Below, we assume that class~$k$ jobs are preemptible by class~$<k$ jobs only.

For all of the alternative models we discuss, performing response time analysis requires (1)~computing the new job joint transforms $\jjt_k$ (\cref{sec:jjt-derivation}) for the system and (2)~determining the formula for the transforms $\lst{X^*_k}$ (\cref{sec:response-time}). Unfortunately, the job joint transforms $\jjt_k$ do not uniquely determine the formula for the transforms $\lst{X^*_k}$.
Finding a natural object from which both of these can be computed remains an open question.
\Cref{sec:extensions:oh} covers alternative overhead models and \cref{sec:extensions:other} covers other systems for which our methods hold.
All but one of the models we consider satisfy the conditions of \cref{sec:JJT-stability:assumptions}, and hence \cref{thm:stability} gives their stability condition.

\subsection{Alternative Overhead Models}\label{sec:extensions:oh}

Our model makes specific choices about overhead dynamics---stochastic pause and resume overheads, nonpreemptible resumes, class-dependent overhead distributions. Each of these choices could reasonably be made differently. Below we show that alternative choices lead to modified job joint transforms that slot into the same framework we use in this paper.

\subsubsection{Preemptible Resumes}

Our current model assumes resume overheads are nonpreemptible, but we could generalize our analysis to a system where resumes can themselves be interrupted by higher-priority arrivals instead of incurring another pause first.
\[
\jjt(\theta, \vec{z}) &= \lst{S_k}(\theta + \lambda_{\geq k}(1 - z_{\geq k}) + \lambda_{< k}(1 - z_{< k} \lst\calO_{k}(\theta, \vec{z}))) \text{, where}\\
\lst\calO_{k}(\theta, \vec{z}) &=
\frac{\lst{C_k}(\theta + \lambda(1 - z))\lst{D_k}(\theta + \lambda_{\geq k}(1 - z_{\geq k}) + \lambda_{< k})}
{1 - \frac{\lambda_{< k}z_{< k}}{\theta + \lambda_{\geq k}(1 - z_{\geq k}) + \lambda_{< k}}\lst{C_k}(\theta + \lambda(1 - z))(1 - \lst{D_k}(\theta + \lambda_{\geq k}(1 - z_{\geq k}) + \lambda_{< k}))}
\]
The transform of $X^*_k$ for this system is quite close to that for our original overhead model, except that the terms from \cref{eq:x-distribution:d, eq:x-distribution:e, eq:x-distribution:g} are all $0$. (Arriving jobs would never have to wait behind a higher-class resume, as it would be preempted.)

\subsubsection{Overhead Dependent on Preempting Job Class.}

In some systems, the overhead incurred by a preempted job may depend on the class of the preempting job as well as the class of the preempted job. Assuming otherwise identical overhead dynamics to our model, this extension would simply require modifying the job joint transform to incorporate class-dependent overhead distributions ($C_{k, i}$ and $D_{k, i}$ to denote the overhead distribution of a class~$k$ job preempted by a class~$i$ job). The transform for the length of an overhead chain now requires solving a $(k-1)$-dimensional system of equations for each class~$k$.
\[
    \jjt(\theta, \vec{z}) = \lst{S_k}(\theta + \lambda_{\geq k}(1 - z_{\geq k}) + \lambda_{< k}(1 - z_{< k}\lst\calO_{k,<k}(\theta, \vec{z})))
    ,
\]
where $\lst\calO_{k, <k}(\theta, \vec{z}) = \sum_{i = 1}^{k - 1}\frac{\lambda_i}{\lambda_{< k}}\lst\calO_{k,i}(\theta, \vec{z})$ and $\lst\calO_{k, i}(\theta, \vec{z})$ is
\[
    &
    \lst{C_{k,i}}(\theta + \lambda(1-z)) \times {}
    \\ & \,\,
        \frac{\lst{D_{k,i}}(\theta + \lambda_{\geq k}(1 - z_{\geq k}) + \lambda_{< k}) + \gp[\big]{\lst{D_{k,i}}(\theta + \lambda_{\geq k}(1 - z_{\geq k}) + \lambda_{< k}) - \lst{D_{k,i}}(\theta + \lambda(1 - z))}\sum_{j\neq i, j< k} \frac{\lambda_j}{\lambda_{< k}}\lst\calO_{k,j}(\theta, \vec{z})}
    {{1 - \frac{\lambda_i}{\lambda_{< k}}\lst{C_{k,i}}(\theta + \lambda(1-z))}\gp[\big]{\lst{D_{k,i}}(\theta + \lambda_{\geq k}(1 - z_{\geq k}) + \lambda_{< k}) - \lst{D_{k,i}}(\theta + \lambda(1 - z))}}
\]
The transform for $X_k^*$ is similar to ours, but the overheads for \cref{eq:x-distribution:b, eq:x-distribution:c, eq:x-distribution:d, eq:x-distribution:e, eq:x-distribution:f, eq:x-distribution:g} are parametrized by $j$ and the class of the preempting job (either $i, \ell,$ or~$k$).

\subsubsection{Jobs with Additional Warmup and Cooldown Phases}

In these systems, jobs incur a warmup time before beginning service and a cooldown period after a job's completion, in addition to the pause and resume overheads incurred by preemptions. We assume these warmups and cooldowns for class~$k$ jobs are nonpreemptible and distributed as $G_k$ and $F_k$ respectively (and follow the same overhead dynamics as our original model), but this can be easily modified.
\[
\jjt(\theta, \vec{z}) = \lst{S_k}(\theta + \lambda_{\geq k}(1 - z) + \lambda_{< k}(1 - z_{< k}\lst\calO_k(\theta,\vec{z}))) \cdot \lst{\calP_k}(\theta,\vec{z}),
\]
where
\[
    \lst\calO_k(\theta, \vec{z})
        &= \frac{\lst{C_k}\gp*{
            \theta + \lambda(1 - z)} \, \lst{D_k}\gp*{
            \theta + \lambda_{\geq k}(1 - z_{\geq k}) + \lambda_{< k}}}
            {1 - \lst{C_k}\gp*{
            \theta + \lambda(1 - z)}  \gp[\big]{
                \lst{D_k}\gp*{
            \theta + \lambda(1 - z)} - \lst{D_k}\gp*{
            \theta + \lambda_{\geq k}(1 - z_{\geq k}) + \lambda_{< k}}
            }}
    \text{ and}\\
    \lst\calP_k(\theta, \vec{z})
        &= \frac{\lst{G_k}\gp*{
            \theta + \lambda(1 - z)} \, \lst{F_k}\gp*{
            \theta + \lambda_{\geq k}(1 - z_{\geq k}) + \lambda_{< k}}}
            {1 - \lst{G_k}\gp*{
            \theta + \lambda(1 - z)}  \gp[\big]{
                \lst{F_k}\gp*{
            \theta + \lambda(1 - z)} - \lst{F_k}\gp*{
            \theta + \lambda_{\geq k}(1 - z_{\geq k}) + \lambda_{< k}}
            }}
    .
\]
Here, $X^*$ has all the same terms as ours, but has additional terms accounting for arriving during another warmup or cooldown.

\subsection{Other Service-Arrival-Dependent Systems}\label{sec:extensions:other}

\subsubsection{Preempt-Repeat Systems}

Preempt-repeat systems have been analyzed in the literature before \cite{asmussen_preemptive-repeat_2017, drekic_reducing_2001, avi-itzhak_preemptive_1963a,gaver_waiting_1962}, but the job joint transform provides a unifying framework for deriving stability and response time results for them.

While our model assumes preempt-resume behavior, some systems employ preempt-repeat mechanisms, where preempted jobs must restart from scratch upon resuming, without any other overhead incurred. In the case of preempt-repeat, the $X^*$ term will always be $0$, as an arriving job will never incur any sort of overhead. Two common preempt-repeat models include:
\* \emph{Preempt-Repeat-Different:} Upon preemption, a job restarts with a new size drawn i.i.d. from the job size distribution. The job joint transform for this model is implicit in the argument of \citet[Proposition~$6$]{asmussen_preemptive-repeat_2017}, and can be explicitly stated as:
\[
\jjt_k(\theta, \vec{z}) = \frac{\lst{S_k}(\theta + \lambda_{\geq k}(1 - z_{\geq k}) + \lambda_{< k})}{1 - \frac{\lambda_{< k}z_{< k}}{\theta + \lambda_{\geq k}(1 - z_{\geq k}) + \lambda_{< k}}\gp[\big]{1 - \lst{S_k}(\theta + \lambda_{\geq k}(1 - z_{\geq k}) + \lambda_{< k})}}
.
\]
\* \emph{Preempt-Repeat-Identical:} Upon preemption, a job restarts with the same size as before. By conditioning on the job size $S_k = s_k$, applying the above preempt-repeat-different result to the \emph{deterministic} size distribution that is always~$s_k$, and then integrating over~$s_k$, we obtain:
\[
\jjt_k(\theta, \vec{z}) = \E[\Bigg]{\frac{\exp(-S_k\cdot(\theta + \lambda_{\geq k}(1 - z_{\geq k}) + \lambda_{< k}))}{1 - \frac{\lambda_{< k}z_{< k}}{\theta + \lambda_{\geq k}(1 - z_{\geq k}) + \lambda_{< k}}\gp[\big]{1 - \exp(-S_k\cdot(\theta + \lambda_{\geq k}(1 - z_{\geq k}) + \lambda_{< k}))}}}
.
\]
\*/

\subsubsection{State-Dependent Arrival Rates}

Another example of a system with service-arrival dependency is the multiclass queue with state-dependent rates, as studied by \citet{ernst_stability_2018}. In this system, the arrival rate of jobs depends upon the class of the job in service. As with preempt-repeat systems, while the system has already been analyzed, our job joint transform can be used to express results for this system.

Let $\lambda_{i, j}$ denote the arrival rate of class~$j$ jobs during the service of a class~$i$ job. Deriving the job joint transform for class-dependent arrival rates is fairly straightforward using the method of collective marks \see\cref{sec:MOCM}. The rate of arrivals of class~$i$ jobs during class~$k$ service is denoted by $\lambda_{k, i}$, so the joint transform of the effective length of class~$k$ service (which does not change based on arrivals during service) and arrivals during that service (which is affected by the class of service) is simply:
\[
\jjt_k(\theta,\vec{z}) = \lst{S_k}\gp[\bigg]{\theta + \sum_{i=1}^n \lambda_{k, i}(1 - z_i)}
.
\]
Note that this system does not satisfy the conditions in \cref{sec:JJT-stability:assumptions}, so the simple stability condition from \cref{thm:stability} does not apply, but it is still the case that the largest eigenvalue of the mean offspring matrix $M = (\E{A_{i, j}})_{i, j}$ determines stability.




\section{Conclusion}\label{sec:conclusion}
This work provides a first transform analysis of response time in the M/G/1 under preemption overhead, specifically under a static preemptive priority scheduling policy. Prior work in scheduling theory has mostly ignored preemption overhead, but our results show that response time under preemption overhead can be analyzed exactly and has a real impact on both response time and stability. We have shown that the response time distribution for static preemptive priority policy under preemption overhead can be expressed via a reduction to an M/G/1/setup system whose parameters are derived from the structure of the \emph{segment tree} \see\cref{sec:tree:response-time}. Our derivation of the recursive Laplace-Stieltjes transform formula \see\cref{sec:response-time} enables exact computation of all response time moments.

We further show that systems with preemption overhead, preempt-repeat, or class-dependent arrival rates share the same underlying service-arrival dependency. All can be analyzed using the \emph{job joint transform} \see\cref{def:jjt}, a general-purpose tool for modeling the two-way dependency between service length and arrival process. Specifically, our stability results from \cref{sec:JJT-stability} apply to any model with service-arrival dependency.

Beyond these specific results, this study serves as a foundation for incorporating preemption overhead into a wider array of scheduling policies. Our formulation is deliberately modular: alternative overhead dynamics can be analyzed by deriving the job joint transform $\lst{\calJ}$ and the extra work distribution $X^*$ \see\cref{def:setup} for the system. This means our approach could apply to a wide class of overhead models, as previewed in \cref{sec:extensions}. A natural next step is to generalize the framework to dynamic priority (age-based) policies and other overhead models, with the goal of developing a comprehensive theory of scheduling under overhead.

\bibliographystyle{ACM-Reference-Format}
\bibliography{refs,more-refs}

\appendix

\section{Notation Table}
\label{app:notation_table}

\BigNotationTable

\section{The Method of Collective Marks}\label{sec:MOCM}

The \emph{method of collective marks} is a probabilistic tool that can be used as a transform analysis technique \citep{runnenburg_van_1979, runnenburg_use_1965}. We employ this method in our transform derivations. Below, we outline its key ideas and terminology.

Let $V$ be a nonnegative continuous random variable, and let $\theta \in \bbR_+$. Let "marks" occur as an independent Poisson point process with rate $\lambda$. The Laplace-Stieltjes transform (LST) of $V$, denoted $\lst{V}(\theta)$, satisfies:
\[
    \lst{V}(\theta) = \E{e^{-\theta V}} = \P{\text{no marks occur during $V$}}.
\]
Thus, $\lst{V}(\theta)$ can be interpreted as the probability that no "marked" Poisson arrivals (occurring at rate $\theta$) occur during $V$. In this context, we refer to $V$ as being "unmarked."

Similarly, for a nonnegative discrete random variable $W$, let $z \in [0, 1]$. Consider sampling $W$ objects, each independently marked with probability $1-z$. The z-transform of $W$, denoted $\widehat{W}(z)$, satisfies:
\[
    \widehat{W}(z) = \E{z^W} = \P{\text{none of the $W$ objects are marked}}.
\]
Thus, $\widehat{W}(z)$ can be interpreted as the probability that no "marks" are encountered among $W$ objects, where each object is independently marked with probability $1-z$. In our transform analysis, we shall describe this as the probability that $W$ is "unmarked."

The ``mark'' intuition extends to joint transforms. For instance, the job joint transform (\cref{def:jjt}):
\[
    \jjt(\theta, z) = \E{e^{-\theta R}z^A},
\]
represents the probability that no Poisson marks, which appear at rate $\theta$, occur during a job's effective service time $R$ and none of the $A$ job arrivals during this effective service time are marked, where each arrival is independently unmarked with probability $z$.

\subsection{Deferred Stability and Busy Period Proofs}\label{sec:JJT-stability:proofs}

\begin{lemma}\label{lem:poisson-standard-joint-tfm}
    The standard Poisson joint transform of a duration distribution~$V$ is
    \[
    \lst{\calP(V)}(\theta, \vec{z}) = \lst{V}(\theta + \lambda_{\leq n}(1 - z_{\leq n})).
    \]
\end{lemma}
\begin{proof}
    Using standard properties of Poisson arrivals \cite[Chapter~27]{harchol-balter_performance_2013}, we compute
    \[
    \lst{\calP(V)}(\theta, \vec{z})
    = \E*{\E*{e^{-\theta V}\prod_{i=1}^n z_i^{A_{V,i}} \given V}}
    =\E*{e^{-\theta V}e^{-V\sum_{i=1}^n\lambda_i(1-z_i)}}  
    = \lst{V}\gp*{\theta + \sum_{i=1}^n \lambda_i(1-z_i)}
    ,
    \]
    and $\sum_{i=1}^n \lambda_i(1 - z_i) = \lambda_{\leq n}(1 - z_{\leq n})$ (see \cref{model:notation}).
\end{proof}

\restate*{thm:busy-period-transform}

\begin{proof}
We use the recursive characterization of busy periods from \cref{def:class-busy-period}. We shall use this and the method of collective marks \see\cref{sec:MOCM} to first prove~(a).

\paragraph{Proof of~(a).}
The transform $\lst{B(\calV)}(\theta)$ gives the probability that no Poisson marks (rate $\theta$) occur during the busy period initiated by a service-arrival joint distribution $\calV$. This busy period consists of the duration component of $\calV$ and the busy periods recursively generated by each class~$i$ job arriving during that time.

The service-arrival joint transform $\lst{\calV}(\theta, \vec{z})$ gives the probability that no Poisson marks occur during the duration component $R_\calV$ (rate $\theta$), and that each class~$i$ arrival is also unmarked with probability $z_i$. We set $z_i = \lst{B(\jobPair_i)}(\theta)$, which is the probability that the busy period initiated by a class~$i$ job is unmarked. Doing so, we obtain the probability that no marks occur during a the root's duration component and no marks occur in the busy period started by each arriving job during this time. We thus achieve the desired result.

\paragraph{Proof of~(b).}
By plugging in $\jobPair_k$ for $\calV$ in (a), it follows that $\vec{b}$ is a nonnegative real solution to the equation implied by (a). It remains to show that it is the least nonnegative solution. This proof closely mirrors that of \citet{ernst_stability_2018} for a more general class of systems with service-arrival dependency.

We begin by defining a "layer-truncated" busy period.

\begin{definition}\label{def:class-busy-period-truncated}
The \emph{$m$-truncated busy period} started by service-arrival joint distribution $\calV$, denoted~$B^{(m)}(\calV)$, is the random variable whose distribution is recursively defined by
\[
B^{(0)}(\calV) \stocheq R_\calV \, 
 \text{ and }\,
B^{(m + 1)} (\calV) \stocheq R_\calV + \sum_{i=1}^{n} \sum_{j=1}^{A_{\calV, i}} \gp[\big]{B^{(m)}(\jobPair_i)}_j\, ,
\]
where $\gp[\big]{B^{(m-1)}(\jobPair_i)}_j$ are i.i.d. copies of $B^{(m)}(\jobPair_i)$.
\end{definition}

Let $b_i^{(m)} = \E*{e^{-\theta B(\jobPair_i)} \, \1(B(\calV) = B^{(m)}(\calV))}$ denote the "$m$-truncated transform" of the busy period, where we "zero out" the contribution from any busy period that exceeds depth~$m$. We write $\vec{b}^{(m)}$ to mean $\sqgp[\big]{b_1^{(m)}, \ldots, b_n^{(m)}}$. By conditioning on the depth of the job tree and using the recursive structure of the busy period, this sequence $\vec{b}^{(m)}$ satisfies:
\[
b_i^{(0)} &= \jobPair_i\gp[\big]{\theta, \vec{0}} \quad \text{and}\\
b_i^{(m+1)} &= \jobPair_i\gp[\big]{\theta, \vec{b}^{(m)}}
\]
for each $m \geq 0$. 

The sequence $\vec{b}^{(m)}$ is well-defined and non-decreasing, since each $\jobPair_i(\theta, \cdot)$ is non-decreasing in its second argument. Moreover, it follows by induction that $\vec{b}^{(m)} \leq \vec{b}$ for all $m$. Since the sequence $\vec{b}^{(m)}$ is monotone increasing and bounded above, the limit $\lim_{m \to \infty} \vec{b}^{(m)}$
exists by the monotone convergence theorem. Taking the limit and using Lebesgue's dominated convergence theorem, we obtain
\[
\lim_{m \to \infty} \vec{b}^{(m)} = \lim_{m \to \infty} \sqgp*{\E*{e^{-\theta B(\jobPair_i)}\, \1(B(\calV) = B^{(m)}(\calV))}}_{i=1}^n = \vec{b}.
\]

Now let $\vec{b}'$ be another arbitrary nonnegative solution to the system, meaning $b_i' = \jobPair_i\gp[\big]{ \theta, \vec{b}' }$ for all~$i$. A similar inductive argument to earlier shows that $\vec{b}^{(m)} \leq \vec{b}'$ for all $m$ as well. Thus we have $\vec{b} \leq \vec{b}'$ pointwise. This implies that $\vec{b}$ is less than or equal to any other nonnegative solution. Therefore, $\vec{b}$ is the least nonnegative solution to the system.
\end{proof}


\restate*{lem:num-arrivals}
\begin{proof}
    Consider the renewal process where an infinite sequence of class~$i$ jobs are served back-to-back while the (time-homogeneous) Poisson arrivals occur.
    Renewal reward theorem implies $\lambda_j = \E{A_{i, j}} / \E{R_i}$.
\end{proof}

\restate*{corr:busy-period-length}

\restate*{thm:class-busy-period-length}

\begin{proof}
Recall from \cref{model:notation} that one can plug in $\ell = n+1$ to obtain the mean of a busy period that admits all classes.
We first show~(a). The expected length of a class~$<\ell$ busy period started by service-arrival joint distribution $\calV$ is simply the sum of the expected effective service time of $\calV$ plus the sum of the \emph{independent} busy periods started by each job arrival of each class $1$ to $\ell -1$ during that duration. By the linearity of expectation and Wald's equation, this simplifies to~(a):
\[
\E*{B_{<\ell}(\calV)} = \E{R_\calV} + \sum_{j=1}^{\ell - 1}\E*{A_{\calV,j}} \E*{B_{<\ell}(\jobPair_j)}.
\]

It remains only to show~(b).
We want to compute the expected length of the class $<\ell$ busy period initiated by a class~$k$ job. Plugging in $\jobPair_k$ for $\calV$ in part (a), we have
\[
\E*{B_{<\ell}(\jobPair_k)} = \E*{R_k} + \sum_{j=1}^{\ell -1} \E*{A_{k,j}} \E*{B_{<\ell}(\jobPair_j)}.
\]
Let $\E{\vec{B}_{< \ell}} = (\E*{B_{<\ell}(\jobPair_1)}, \ldots, \E*{B_{<\ell}(\jobPair_{\ell - 1})})^\top$ denote the vector of expected class $<\ell$ busy period lengths, and let $\E{\vec{R}_{< \ell}} = (\E*{R_1}, \ldots, \E*{R_{\ell - 1}})^\top$ denote the vector of expected effective service times for all classes~$<\ell$. Let $M_{< \ell}$ be the $(\ell-1) \times (\ell-1)$ mean matrix with entries $M_{i,j} = \E*{A_{i,j}}$ (as in \cref{thm:stability}).

From \cref{lem:num-arrivals}, we have $M_{i,j} = \lambda_j \E*{R_i}$, so $M_{< \ell}$ is a rank-one matrix of the form $\E{\vec{R}_{< \ell}} \vec{\lambda}_{< \ell}^\top$, where $\vec{\lambda}_{< \ell} = (\lambda_1, \ldots, \lambda_{\ell - 1})^\top$ is the vector of class~$< \ell$-specific arrival rates. Then it follows that
\[
\E{\vec{B}_{< \ell}} = \E{\vec{R}_{< \ell}} + M_{< \ell} \E{\vec{B}_{< \ell}} = \E{\vec{R}_{< \ell}} + (\E{\vec{R}_{< \ell}} \vec{\lambda}_{< \ell}^\top) \E{\vec{B}_{< \ell}}.
\]
This implies
$(I - M_{< \ell})\E{\vec{B}_{< \ell}} = \E{\vec{R}_{< \ell}}$.
Since $M_{< \ell}$ is a rank-$1$ matrix, we can compute the following inverse using the Sherman-Morrison formula, obtaining
\[
(I - \E{\vec{R}_{< \ell}} \vec{\lambda}_{< \ell}^\top)^{-1} = I + \frac{\E{\vec{R}_{< \ell}} \vec{\lambda}_{< \ell}^\top}{1 - \vec{\lambda}_{< \ell}^\top \E{\vec{R}_{< \ell}}},
\]
provided $1 - \vec{\lambda}^\top \E{\vec{B}_{< \ell}} \neq 0$. In our case, $\vec{\lambda}_{< \ell}^\top \E{\vec{B}_{< \ell}} = \sum_{j<\ell} \lambda_j \E*{R_j} = \rho_{<\ell} \leq \rho < 1$ by the stability condition.
Therefore, for any class $k < \ell$, we find that
\[
\E*{B(\jobPair_k)} = e_k^\top \E{\vec{B}_{< \ell}} = e_k^\top (I - M_{< \ell})^{-1} \E{\vec{R}_{< \ell}} = \frac{\E*{R_k}}{1 - \rho_{<\ell}}.
\]
\end{proof}

\section{Proofs for Response Time Transform}\label{app:x-proofs}
\restate*{lem:x-distribution}
\begin{proof}
    We wish to find the distribution of $X^*$, the length of time between an early class~$k$ job's arrival and when any class~$k$ job can next begin service. To do this, we must consider the work a early class~$k$ job encounters upon arriving to the preemptive priority system with overhead. Recall that a class~$k$ job being an early arrival means it arrives when no class~$k$ in the system have received any service thus far \see\cref{def:early}.

    Thus the system can be in only one of the following states when an early class~$k$ job arrives to the system:
    \*[topsep=\medskipamount, itemsep=\medskipamount][\itshape(a--c)] The job arrives while the system is completing strictly higher priority work; specifically, is in the middle of a class~$<k$ busy period started by:
    \**[\bfseries(a)] A class~$< k$ job that arrived to an idle system.
    \** The pause overhead triggered by an earlier job that arrived during service of a class~$> k$ job (plus the earlier job, if it was of a class less than~$k$).
    \** The pause overhead triggered by an earlier job arriving during a class~$> k$ resume  (plus the earlier job, if it was of a class less than~$k$).
    \*[\itshape(d--e)] The job arrives while the system is completing strictly lower-priority work, and this job's arrival triggers overhead, specifically:
    \**[\bfseries(a), resume] The job arrives during service of a class~$> k$ job and triggers pause overhead.
    \** The job is the first to arrive during the resume of a class~$> k$ job and causes that resume to fail, triggering pause overhead
    \*[\bfseries(f)] The job arrives while the system is completely idle.
    \*/
    This observation directly gives us the distribution of each state, except those of (c) and (e).
    We begin by computing the above distribution of the states for (c) and (e).

    For state (c), the work an early class~$k$ job sees is the remainder of the busy period started by the resume, the resulting pause, and the job that caused the resume to fail (if it was of class~$< k$). Observe that we are conditioning on the fact that the resume failed, meaning its length is biased. We will call the length of a resume conditioned on it failing $D^{(1)}_j$. We compute the transform for this amount in \cref{lem:resume-tfms}.

    For state (e), we note that if the tagged job itself caused the resume to fail, the job will experience waiting time equal to the remainder of that resume, plus the pause it triggers, plus the class~$<k$ busy period of that pause.  Because we are conditioning on the tagged job being the first to arrive to this resume, the length of the remainder, $(D^{(0)}_j)_\e$, is biased. We compute the transform for this amount in \cref{lem:resume-tfms}.
    
    It remains to find the probability of encountering each state. Because we are looking at the distribution of work seen by an early class~$k$ job, this means that the server has not yet begun working on any class~$k$ jobs in the system. Recall from \cref{sec:tree:response-time} that this means no class~$k$ superjobs are in-progress when the an early arrives. That is, the server is class~$k$ superjob-idle. Recall that class~$k$ superjobs have length distributed as $B_{<k}(\jobPair_k)$. Because (a)-(f) cannot occur if the system is currently working on a class~$k$ superjob, the probability of an early class~$k$ job encountering each state (a)-(f) is equal to the probability of that state occurring divided by the probability that the server is class~$k$ superjob-idle:
    \[
    \P{\text{server is class~$k$ superjob-idle}}&= 1 - \lambda_k \E{B_{<k}(\jobPair_k)} \\ &=
    1 - \frac{\lambda_k\E{R_k}}{1-\rho_{<k}} \text{\, by \cref{corr:busy-period-length}}\\ &= \frac{1-\rho_{<k}}{1-\rho_{<k}} - \frac{\rho_k}{1-\rho_{<k}} \\
    &= \frac{1-\rho_{\leq k}}{1-\rho_{<k}}.
    \]
    
    \paragraph{State (a): The class~$k$ job arrives while the system is completing a busy period started by a class~$< k$ job that arrived to an idle system}
    By Little's Law \cite{little_littles_2011}, the probability of encountering the system in such a state is the rate at which such busy periods occur times the expected length of these busy periods. The rate of these busy periods is $\lambda_i(1-\rho)$, and the expected length of such a busy period is $\frac{\E{R_i}}{1-\rho_{<k}}$ by \cref{corr:busy-period-length}.
    
    \paragraph{State (b): The class~$k$ job arrives while the system is completing a busy period started by the overhead triggered by an earlier job that arrived during service of a class~$> k$ job}
    The probability of encountering the system in such a state can be split into two parts:
    
    In the case that a class $i < k$ job caused the pause, the probability of encountering a busy period started by the resulting overhead can be computed using Little's Law. The rate of pause overheads caused by a class~$i$ job arriving during class~$j$ service is $\lambda_i\sigma_j$ by Poisson splitting. The expected length of the resulting supersetup is the length of a class~$<k$ busy period started by a class~$j$ pause and the class~$i$ job that caused the pause, which is just $\frac{\E{C_j + R_i}}{1 - \rho_{<k}}$ by \cref{corr:busy-period-length}.

    In the case that a class~$\ell \in [k,j)$ job caused the pause, the probability of encountering this state is similar to the above, except that $\lambda_i$ becomes $\lambda_{\ell}$ and the expected length of the resulting supersetup is just $\frac{\E{C_j}}{1 - \rho_{<k}}$ because the job that caused the pause is either not higher-priority than a class~$k$ job, or is a class~$k$ job and therefore has been accounted for elsewhere.
    
    \paragraph{State (c): The class~$k$ job arrives while the system is completing a busy period started by the overhead triggered by an earlier job arriving during a class~$> k$ resume}
    The probability of encountering the system in such a state can be split into two parts:
    
    In the case that a class $i < k$ job caused a failed resume, the probability of encountering a busy period started by this can be computed using Little's Law.\footnote{To understand why we can break down this probability using Little's Law in the following way in more detail, see \cref{app:conditional-prob}.} The rate of failed resumes caused by a class~$i$ job, by the independence of resume durations and the Palm Inversion Formula, is equal to the probability that the first class~$<j$ arrival in a resume was class~$i$, times the rate of overhead chains times the expected number of resumes that fail per overhead chain. The probability that the first class~$<j$ arrival in a resume was class~$i$ is $\frac{\lambda_i}{\lambda_{<j}}$. The rate of overhead chains for a class~$j$ job is $\lambda_{<j}\sigma_j$. The expected number of resumes that fail per overhead chain is equal to the number of resumes per overhead chain minus one, which is just $\frac{1 - \lst{D_j}(\lambda_{<j})}{\lst{D_j}(\lambda_{<j})}$. The expected length of the resulting supersetup is the length of a class~$<k$ busy period started by a failed resume, the class~$i$ job that caused the resume, and the class~$j$ pause resulting from the failed resume, which is just $\frac{\E{D^{(1)}_j + C_j + R_i}}{1 - \rho_{<k}}$ by \cref{corr:busy-period-length}.

    In the case that a class~$\ell \in [k,j)$ job caused the resume to fail, the probability of encountering this state is similar to the above, except that $\lambda_i$ becomes $\lambda_{\ell}$ and the expected length of the resulting supersetup is just $\frac{\E{D^{(1)}_j + C_j}}{1 - \rho_{<k}}$ because the job that caused the resume to fail is either not higher-priority than a class~$k$ job, or is a class~$k$ job and therefore has been accounted for elsewhere.
    
    \paragraph{State (d): The class~$k$ job arrives during service of a class~$> k$ job and triggers pause overhead}
    In the instant before this job arrives, there is no work relevant to class~$k$ jobs in the system. However, once this job arrives, it triggers a class~$> k$ pause overhead and will have to wait for the class~$< k$ busy period started by this pause to conclude before any class~$k$ job can begin service.
    The probability of an early class~$k$ job encountering the system in this state is the load of class~$> k$ service.
    
    \paragraph{State (e): The class~$k$ job is the first to arrive during the resume of a class~$> k$ job and causes that resume to fail, triggering pause overhead}
    In the instant before this job arrives, the only work relevant to class~$k$ jobs in the system is the \emph{remainder} of the thus-far-uninterrupted class~$> k$ resume in progress. However, once this job arrives, it triggers a class~$> k$ pause overhead and will have to wait for the class~$< k$ busy period started by the remainder of this resume and the resulting pause to conclude before any class~$k$ job can begin service. This length of work is distributed as the excess of the resume conditioned on it being uninterrupted before our arrival, plus the class~$< k$ busy period started by the resulting pause.
    The probability of an early class~$k$ job encountering the system in this state is the long-run probability of arriving during a resume and that there have been no arrivals before the time at which you arrive. The long-run average probability of arriving during a resume is $\delta_j$. If we consider a renewal process of resumes, the probability that our job is the first arrival in a resume is the probability that the age of the resume during which we arrive has no class~$<j$ arrivals. This is equal to $(\lst{D_j})_\e(\lambda_{< j})$ by the method of collective marks \see\cref{sec:MOCM}. We can evaluate this as:
    \[
    \delta_j\cdot (\lst{D_j})_\e(\lambda_{< j}) &= \delta_j\cdot \frac{1 - \lst{D_j}(\lambda_{<j})}{\lambda_{<j}\E{D_j}} \\
    &= \delta_j\cdot \frac{1 - \lst{D_j}(\lambda_{<j})}{\lambda_{<j}\E{D_j}} \cdot \frac{\sigma_j}{\sigma_j} \cdot \frac{\lst{D_j}(\lambda_{<j})}{\lst{D_j}(\lambda_{<j})} \\
    &= \sigma_j \gp*{\frac{1 - \lst{D_j}(\lambda_{<j})}{\lst{D_j}(\lambda_{<j})}}.
    \]
    
    \paragraph{State (f): The class~$k$ job arrives to an idle system}
    In this case, there is $0$ supersetup work a class~$k$ job must wait behind before beginning service. The probability of an early class~$k$ job encountering the system in such a state is the probability of an arrival encountering an idle system.
    
    From this accounting, the desired formulas follow.
\end{proof}

\restate*{lem:resume-tfms}
\begin{proof}
    Both follow from a straightforward application of the method of collective marks (\cref{sec:MOCM}).
    The transform of a class~$j$ resume conditioned on it having failed is equal to the probability that a class~$j$ resume failed and had no marks divided by the probability that it failed. The probability that a class~$j$ resume failed and had no marks is equal to the probability that it had no marks minus the probability that it succeeded and had no marks. Hence the desired formula is:
    \[
    \lst D_j^{(1)}(\theta) &= \frac{\lst{D_j}(\theta) - \lst{D_j}(\theta + \lambda_{< j})}{1 - \lst{D_j}(\lambda_{< j})}.
    \]
    The transform of $D_j^{(0)}$ follows from the transform of $D_j^{(1)}$ and the fact that $\lst{D_j}(\lambda_{<j})$ is the probability of a resume succeeding. Specifically,
    \[
    \lst{D_j}(\lambda_{<j})\cdot \lst{D_j^{(0)}}(\theta) + (1 - \lst{D_j}(\lambda_{<j}))\cdot \lst{D_j^{(1)}}(\theta) = \lst{D_j}(\theta)
    \]
\end{proof}

\section{Distribution of Resumes}\label{app:conditional-prob}
When we are computing the long-run probability of an early class~$k$ job encountering a class~$<k$ busy period tree started by a class~$< j$ job arriving during a class~$j > k$ resume and causing a pause, we decompose it using Little's Law, and condition the length of this resume on it having failed. We can do this because we are effectively first sampling whether or not a resume has failed, then computing the length of the resume conditioned on it having failed. The length of the busy period started by this resume is just the length of a busy period started by a resume that has failed, because the time of the first arrival, conditioned on the resume having failed, is uniformly distributed. In the case where the job causing the resume to fail was of class~$i < k$, the reason we can separately count this class~$i$ job that caused the pause without modifying the busy period distribution is because the arrival process is Poisson.

\end{document}